\documentclass[twoside,11pt]{article}

\usepackage{blindtext}

 \usepackage[abbrvbib, preprint]{jmlr2e}

\usepackage{amsmath}
\usepackage{graphicx}
\usepackage{subcaption}  
\usepackage{algorithm}
\usepackage{algpseudocode}
\usepackage{booktabs}       

\usepackage[most]{tcolorbox}

\DeclareMathOperator*{\argmin}{arg\,min}

\newtheorem{assumption}{Assumption}

\usepackage{lastpage}
\jmlrheading{-}{2025}{1-\pageref{LastPage}}{11/22; Revised 5/31}{-/-}{24-2008}{Hang Ren, Yulin Wu, Shuhan Qi, Jiajia Zhang, Xiaozhen Sun, Tianzi Ma and Xuan Wang}

\ShortHeadings{Regret Minimization via Adaptive Reward Transformation}{Ren, Wu, Qi, Zhang, Sun, Ma and Wang}
\firstpageno{1}

\begin{document}

\title{Efficient Last-Iterate Convergence in Regret Minimization via Adaptive Reward Transformation}

\author{
    \name Hang Ren$^{1}$ \email renhang@stu.hit.edu.cn \\
    \name Yulin Wu$^{1,2}$ \email yulinwu@cs.hitsz.edu.cn \\
    \name Shuhan Qi$^{1,2}$ \email shuhanqi@cs.hitsz.edu.cn \\
    \name Jiajia Zhang $^{1,2}$ \email zhangjiajia@hit.edu.cn \\
    \name Xiaozhen Sun$^{1}$ \email sxzlsk515@outlook.com \\
    \name Tianzi Ma$^{1}$ \email tianzima@stu.hit.edu.cn \\
    \name Xuan Wang$^{1,2}$\thanks{Corresponding author.} \email wangxuan@cs.hitsz.edu.cn \\
    \addr $^1$School of Computer Science and Technology, Harbin Institute of Technology (Shenzhen), China \\
    \addr $^2$Guangdong Provincial Key Laboratory of Novel Security Intelligence Technologies
}

\editor{Tor Lattimore}
\maketitle
\begin{abstract}
Regret minimization is a powerful method for finding Nash equilibria in Normal-Form Games (NFGs) and Extensive-Form Games (EFGs), but it typically guarantees convergence only for the average strategy. However, computing the average strategy requires significant computational resources or introduces additional errors, limiting its practical applicability. The Reward Transformation (RT) framework was introduced to regret minimization to achieve last-iterate convergence through reward function regularization. However, it faces practical challenges: its performance is highly sensitive to manually tuned parameters, which often deviate from theoretical convergence conditions, leading to slow convergence, oscillations, or stagnation in local optima.

Inspired by previous work, we propose an adaptive technique to address these issues, ensuring better consistency between theoretical guarantees and practical performance for RT Regret Matching (RTRM), RT Counterfactual Regret Minimization (RTCFR), and their variants in solving NFGs and EFGs more effectively. Our adaptive methods dynamically adjust parameters, balancing exploration and exploitation while improving regret accumulation, ultimately enhancing asymptotic last-iterate convergence and achieving linear convergence. Experimental results demonstrate that our methods significantly accelerate convergence, outperforming state-of-the-art algorithms.
\end{abstract}


\section{Introduction}
Game theory provides essential frameworks for solving real-world problems, with Normal-form Games (NFGs) and Extensive-form Games (EFGs) serving as key models for studying strategic interactions, particularly in scenarios involving imperfect information and sequential decisions. A primary objective in game theory is to find or approximate Nash equilibria. Over time, numerous algorithms have been proposed, but the most efficient approaches are based on no-regret learning, valued for their simplicity, parameter-free nature, and effectiveness. This foundation underlies Regret Matching (RM) \citep{hart2000simple} in NFGs, as well as its extensive-form counterparts—Counterfactual Regret Minimization (CFR) \citep{zinkevich2007regret} and CFR+ \citep{tammelin2014solving}—which leverage these properties to approximate Nash equilibria via the average strategy.

Despite the advantages of average-strategy convergence, it requires considerable computational resources and introduces additional errors when approximating the average strategy, limiting its applicability in large-scale games. For instance, players may be unable to recall all past decisions, and deep learning-based methods struggle to effectively average historical strategies or parameters, complicating the integration of neural networks in these solutions. This limitation has led to the use of additional neural networks in Deep CFR \citep{brown2019deep} to train the average strategy, which in turn introduces further errors.

These limitations strongly motivate the development of regret minimization algorithms that achieve convergence using only the last strategy—a concept known as last-iterate convergence. Ideally, last-iterate convergence should be faster than that of the average strategy, as averaging inherently incorporates suboptimal strategies, which can hinder efficient convergence to an optimal solution. However, common methods such as regret matching, regret matching+, and Hedge often fail to achieve this property both theoretically and empirically, even in normal-form games \citep{farina2024regret}. Although empirical studies have shown that CFR can achieve better performance with its last strategy compared to its average \citep{lockhart2019computing}, and \citet{cai2023last} observed rapid last-iterate convergence in certain cases with extragradient RM+ and predictive RM+, a comprehensive theoretical framework explaining these phenomena in EFGs is still lacking.

Recent studies have demonstrated that Optimistic Gradient Descent Ascent (OGDA), Optimistic Multiplicative Weights Update (OMWU), and regularized counterfactual regret minimization (Reg-CFR) achieve last-iterate convergence in two-player zero-sum games, both in normal and extensive forms \citep{wei2020linear, lee2021last, liu2022power}. However, these methods face significant limitations, including the need for careful hyper-parameter tuning and complex gradient calculations, which can hinder their practical effectiveness in large-scale games, with some even failing to converge reliably.

An alternative approach, inspired by evolutionary game theory, is the Reward Transformation (RT) framework \citep{bauer2019stabilization, perolat2021poincare, perolat2022mastering}. This method, when applied to Follow the Regularized Leader (FTRL) and Multiplicative Weights Update (MWU) in continuous-time feedback games, has shown results comparable to or better than those of OGDA and OMWU \citep{abe2022mutation, abe2022last}, although its performance in real-world applications remains suboptimal. Recently, the RT framework has been adapted to RM+ and CFR+ for solving discrete-time feedback games, achieving last-iterate convergence in complex games for the first time \citep{meng2023efficient}. However, it has also encountered several challenges: the introduction of RT-term parameters, which are usually manually tuned and set to be fixed in practice, deviates from convergence theory, and performance is highly sensitive to these parameters, often leading to slow convergence, oscillations, or stagnation at local optima.

In this work, inspired by the work of \citet{perolat2021poincare} and \citet{meng2023efficient}, we propose adaptive RT regularization techniques for regret minimization algorithms, including RM, CFR, and their variants, to achieve efficient last-iterate convergence for finding Nash equilibrium in two-player zero-sum games. Our main contributions are as follows:

\begin{itemize}
    \item We extend the convergence analysis of the RT framework beyond \citet{meng2023efficient}'s RTRM+ to general RTRMs algorithms, proving a best-iterate rate of \(O(1/\sqrt{T})\) and asymptotic last-iterate convergence. We identify unresolved bottlenecks in \citet{meng2023efficient}, including reference strategy selection and fixed RT-term and regret weights, which cause instability and hinder convergence, thus improving the framework's robustness.
    
    \item We propose an adaptive technique to address these bottlenecks from three aspects: (1) dynamically selecting reference strategies, (2) adjusting RT weights to either accelerate convergence or explore new optimization directions, and (3) improving the regret accumulation process by introducing adaptive discounted regret weighting \citep{brown2019solving}, which better aligns with achieving last-iterate convergence.

    \item We apply the adaptive RT framework to EFGs using laminar regret decomposition and prove its asymptotic last-iterate convergence. This addresses the technical challenges in guaranteeing the convergence of RTCFR+ as established in~\citet{meng2023efficient}.

    \item We evaluate our adaptive methods against state-of-the-art average-iterate and last-iterate convergence algorithms in both NFGs and EFGs, demonstrating significant performance improvements.
\end{itemize}

The structure of the paper is as follows: In Section 2, we discuss related work on equilibrium-finding algorithms, average-iterate convergence, and last-iterate convergence algorithms. In Section 3, we introduce preliminaries on NFGs and EFGs, the details of regret minimization, and the RT framework. Section 4 presents the convergence analysis of the RT framework and identifies its bottlenecks, which we address with our proposed adaptive method in Section 5, including parameter selection and RM structure optimization for consistency with last-iterate convergence. In Section 6, we extend our work to CFRs for solving EFGs with imperfect information, followed by experimental evaluation in Section 7. Finally, we conclude this paper and discuss future work in Section 8.

\section{Related Work}

In this section, we focus on two-player zero-sum games with perfect recall, commonly known as saddle-point optimization problems. The concept of Nash equilibrium \citep{nash1950equilibrium, nash1950non} and approximate Nash equilibrium has driven the development of numerous algorithms aimed at finding the optimal strategy profile, also referred to as the fixed point in saddle-point optimization. Notable methods include linear programming \citep{von1996efficient, koller1994fast}, first-order methods \citep{nesterov2005excessive, hoda2010smoothing, kroer2015faster, kroer2017theoretical, kroer2018solving}, fictitious play \citep{brown1951iterative, heinrich2015fictitious, heinrich2016deep}, the double oracle approach \citep{mcmahan2003planning}, convex optimization techniques \citep{kalai2005efficient, hazan2016introduction}, and Counterfactual Regret Minimization (CFR) \citep{zinkevich2007regret, bowling2015heads, lanctot2009monte}. Since strategies in these games can be represented in both behavior and sequence-forms \citep{von1996efficient}, many of these algorithms are interconnected \citep{waugh2015unified, farina2021faster, liu2022equivalence}, leading to significant advancements in game-theoretic strategies.

\subsection{Average-Iterate Convergence in Saddle-Point Optimization} 
Average-iterate convergence algorithms, such as Mirror Descent with Euclidean function regularization in Gradient Descent Ascent (GDA) and entropy function regularization in Multiplicative Weights Update (MWU), along with Follow-the-Regularized-Leader (FTRL) and CFR \citep{zinkevich2007regret}, achieve a convergence rate of \(O(1/\sqrt{T})\), where \(T\) represents the number of iterations. These algorithms rely on the average strategy, necessitating the computation and storage of the average strategy at each step, which increases both computational and memory demands.

Moreover, historical strategies can significantly impact outcomes, especially when earlier strategies diverge substantially from the optimal or contain errors—an issue often encountered in algorithms developed from scratch. Such reliance on historical strategies extends the influence of earlier, potentially suboptimal decisions, which can complicate convergence to the optimal strategy.

\subsection{Last-Iterate Convergence in Saddle-Point Optimization} 
Recent advancements in Optimistic Gradient Descent Ascent (OGDA) and Optimistic Multiplicative Weights Update (OMWU) have shown potential for achieving last-iterate convergence in saddle-point optimization. These methods have been effective in both NFGs \citep{wei2020linear} and EFGs \citep{lee2021last}. However, GDA typically incurs significant computational overhead due to gradient calculations and projections, while MWU offers a closed-form solution with reduced computational demands but requires the assumption of a unique saddle point, which can be restrictive.

To address these challenges, recent studies \citep{anagnostides2022last, liu2022power} have modified regularization techniques to eliminate the need for the unique saddle point assumption, thereby achieving last-iterate convergence. However, the theoretical hyperparameters required for convergence are very small, leading to slow convergence in practice. Increasing these hyperparameters empirically often results in poorer performance, even in specific applications such as Leduc poker with three ranks \citep{wei2020linear}.

In evolutionary game theory, reward transformation techniques have been applied to FTRL and MWU methods \citep{perolat2021poincare, perolat2022mastering, abe2022mutation, abe2022last} for solving continuous-time feedback games. These methods do not rely on the assumption of a unique saddle point. Instead, they introduce a ``mutant" term in the reward function, which accounts for the difference between the reference strategy and the last strategy. This adjustment helps to make the saddle point an attractor of the dynamics, guiding strategy updates toward a proximate Nash equilibrium. \citet{meng2023efficient} extended this work by applying it to RM+ and CFR+, transitioning from continuous to discrete game problems. This body of work forms the basis of the Reward Transformation (RT) framework, where convergence guarantees are established by constructing a Lyapunov function related to the distance between the saddle point and the last strategy. This ensures asymptotic last-iterate convergence, although the specific convergence rate remains undefined.

We have identified that the main bottleneck of the RT framework is the convergence rate and stability, which is highly sensitive to parameters such as the reference strategy and the RT weight. In the worst case, the algorithm may fail to converge in practice. This motivates further study into how adaptive parameter selection during iterations can overcome this bottleneck and improve both the theoretical and practical performance of the RT framework.

\section{Preliminaries}
\subsection{Basic Notation} 
The standard inner product of vectors \(x\) and \(y\) is denoted by \(\langle x, y \rangle\). For a vector \(x \in \mathbb{R}^n\), we define its \(l_p\)-norm as \(\|x\|_p := \left( \sum_{i=1}^n |x_i|^p \right)^{1/p}\) for \(p \in [1, \infty)\). To measure the difference between two vectors, we use the Bregman distance \(D_\psi(x, y)\), which is defined using a distance-generating function \(\psi(\cdot)\):
\[
D_\psi(x, y) = \psi(x) - \psi(y) - \nabla \psi(y) \cdot (x - y).
\]
In this paper, we set \(\psi(x) = \frac{1}{2}\|x\|_2^2\), resulting in the specific form \(D_\psi(x, y) = \frac{1}{2}\|x - y\|_2^2\). We omit the notation if not explicitly specified.

\subsection{Normal-Form Game and Extensive-Form Game}

A \emph{Normal-Form Game (NFG)} is defined as a tuple \(G = (\mathcal{X}_1, \mathcal{X}_2, u)\), where \(\mathcal{X}_1 \subseteq \mathbb{R}^n\) and \(\mathcal{X}_2 \subseteq \mathbb{R}^m\) represent the convex and compact action spaces for player 1 and 2, respectively. The utility function \(u_i: \mathcal{X}_1 \times \mathcal{X}_2 \rightarrow \mathbb{R}\) is a biaffine mapping that assigns payoffs to player \(i\) based on the action pair \((x_1, x_2) \in \mathcal{X}_1 \times \mathcal{X}_2\). In zero-sum games, it holds that \(u_2(x_1, x_2) = -u_1(x_1, x_2)\).

An \emph{Extensive-Form Game (EFG)} is defined as a tuple \(G = (H, Z, A, P, \mathcal{I}_i, \sigma_i, u_i)\). In this representation, \(H\) is the set of states, including the initial state \(\emptyset\), while \(Z \subseteq H\) represents the set of terminal states or leaf nodes. The actions available at a non-terminal state \(h \in H \setminus Z\) are denoted by \(A(h)\). The player function \(P: H \to \{0, 1, 2\}\) maps each non-terminal history \(h \in H \setminus Z\) to the player who is to move at that state; if \(P(h) = 0\), the player is the ``chance" player.

The set \(\mathcal{I}_i\) represents the information partition for player \(i \in \{1, 2\}\), where each information set \(I \in \mathcal{I}_i\) consists of states \(h \in H\) such that \(P(h) = i\), which are indistinguishable to player \(i\). For any \(h, h' \in I\), it holds that \(A(h) = A(h')\). For convenience, we denote \(A(I)\) to mean \(A(h)\) for any \(h \in I\).

The probability of each action \(a\) at information set \(I\), controlled by player \(i\) (where \(P(I) = i\)), is given by \(\sigma_i: \mathcal{I} \times A \to [0, 1]\). This function follows  standard simplex property, defined as \( \sigma_i(I) \in \Delta^{|A(I)|} := \{ x \in \mathbb{R}^{|A(I)|} : x \geq 0, \sum_{a \in A(I)} x(a) = 1 \} \), which we refer to as the behavior strategy. The utility function for player \(i \in \{1, 2\}\) at each terminal state \(z \in Z\) is denoted by \(u_i: Z \to \mathbb{R}\). In a two-player zero-sum game, it also holds that \(u_1(z) = -u_2(z)\).

A \emph{sequence} \(s=Ia\) represent the history where the player has taken a series of actions to arrive at information set \(I\) and then applies action \(a\). Due to the assumption of perfect recall, this history is unique. The set of \emph{sequences} for player \(i\) is \(\mathcal{S}_i=\{\emptyset\}\cup\{Ia:I\in \mathcal{I}_i,a\in A(I)\}\), where \(\emptyset\) is the \emph{empty sequence}. For an information set \(I\in \mathcal{I}_i\),  The \emph{parent sequence} \(pI \in S_i\) is the last sequence on the path from the root to \(I\), or \(\emptyset\) if player \(i\) has not acted before \(I\). The game tree induces a partial order over states, information sets, and sequences, denoted by \(\prec, \preceq\). For example, \(I'a' \prec Ia\) indicates that the path to action \(a\) at \(I\) passes through action \(a'\) at \(I'\); \(I'a' \preceq Ia\) means \(I'a' \prec Ia\) or \(I'a' = Ia\); and \(I' \preceq I\) means \(I\) is a descendant of \(I'\) or \(I' = I\). The set \(C(I, a)\subset \mathcal{I}_i\) represents the information sets that can be immediately reached when player \(i\) applies action \(a\) in information set \(I\).

The sequence of actions corresponds to each action in every information set, and a probability distribution over these sequences represents the \emph{sequence-form strategy}, represented by a vector \(q_i\in \mathcal{Q}_i \subseteq \mathbb{R}^{|\mathcal{S}_i|}\). This formulation maps one-to-one to the behavior strategy and ensures that the representation remains linear with respect to the size of the game tree, unlike the mixed strategy, which may lead to combinatorial explosion. The sequence-form strategy \(q_i\) is computed as the product of the behavioral strategy \(\sigma_i\) for player \(i\), applied to each action along the sequence:
\begin{equation}
    q_i(Ia) = \prod_{I'a'\in \mathcal{S}_i: I'a' \preceq Ia} \sigma_i(I', a'),
    \label{eq:q}
\end{equation}
with \(q_i(\emptyset)=1\). Conversely:
\begin{equation}
    \sigma_i(I,a)=\frac{q_i(Ia)}{q_i(pI)}
    \label{eq:x}
\end{equation}
where \(q_i(pI)=\sum_{a\in A(I)}q_i(Ia)\). Equation~\eqref{eq:q} and \eqref{eq:x} enable conversion between sequence-form and behavior strategies.

We also define \(q_i(h) = \prod_{h' \preceq h} \sigma_i(h', a')\) as the reach probability of player \(i\) at state \(h\) along that history. Using this representation allows us to treat EFGs similarly to NFGs: \(G=(\mathcal{Q}_1, \mathcal{Q}_2, U)\) , facilitating the search for a Nash equilibrium (NE) in two-player zero-sum games. This is equivalent to finding the saddle point of a convex-concave optimization problem \citep{von1996efficient}:

\begin{equation}
    \min_{q_1 \in Q_1} \max_{q_2 \in Q_2} \langle q_1, -U q_2 \rangle = \max_{q_2 \in Q_2} \min_{q_1 \in Q_1} \langle q_1, -U q_2 \rangle
\label{eq: minmax}
\end{equation}

Here, \(q_1 \in Q_1 \subseteq \mathbb{R}^{|\mathcal{S}_1|}\) and \(q_2 \in Q_2 \subseteq \mathbb{R}^{|\mathcal{S}_2|}\) represent the sequence-form strategies of the players. The term \(|\mathcal{S}|\) denotes the number of sequences for each player, while \(U \in \mathbb{R}^{|\mathcal{S}_1| \times |\mathcal{S}_2|}\) represents the utility matrix for player 1. The game tree can be constructed as treeplexes \citep{hoda2010smoothing}, allowing for an efficient representation of the strategy space.

We use a metric called ``exploitability" to measure how far a strategy \(q = (q_1, q_2)\) is from the Nash equilibrium:
\begin{equation}
     \epsilon(q) := \max_{q_1' \in Q_1} \langle q_1', U q_2 \rangle - \min_{q_2' \in Q_2} \langle q_1, U q_2' \rangle
\end{equation}
Exploitability also represents an approximate \(\epsilon\)-NE, meaning that a strategy with zero exploitability corresponds to an exact NE of the game.

\subsection{Regret Matching and Counterfactual Regret Minimization}
Regret Matching (RM) \citep{hart2000simple} and Regret Matching+ (RM+) are widely used regret minimization algorithms for finding Nash equilibria in NFGs. For a two-player game, let \( U_i \in \mathbb{R}^{|A_1| \times |A_2|} \) denote the utility matrix for player \( i \), where \( A_1 \) and \( A_2 \) are the action sets of players 1 and 2, respectively. The index \( -i \) denotes the opponent of player \( i \), and \( \sigma_{-i} \in \Delta^{|A_{-i}|} \) represents the opponent's strategy. The dynamics of RM and RM+ follow these four steps in each iteration: compute the action loss (negative utility), immediate regret, cumulative regret, and next strategy, as follows:

\begin{align}
    \ell^t_i & = -U_i \sigma^t_{-i} \label{v} \\
    r^t_i & = \langle \ell^t_i, \sigma^t_i \rangle \mathbf{1} - \ell^t_i \label{r} \\
    R^t_i & = \begin{cases}
         R^{t-1}_i + r^t_i, & \text{if RM} \\
         [R^{t-1}_i + r^t_i]_+, & \text{if RM+}
    \end{cases} \label{R} \\
    \sigma^{t+1}_i & = \begin{cases}
        \frac{\left[R^t_i\right]_+}{\|\left[R^t_i\right]_+\|_1}, & \text{if } \|\left[R^t_i\right]_+\|_1 > 0 \\
        \frac{1}{|A_i|}, & \text{otherwise}
    \end{cases} \label{x}
\end{align}
The primary distinction between RM and RM+ is that RM+ ensures non-negative cumulative regret by applying the projection operator \([x]_+ = \max(0, x)\).

Counterfactual Regret Minimization (CFR) \citep{zinkevich2007regret} and CFR+ \citep{tammelin2014solving} are equilibrium-finding algorithms for EFGs that minimize regret within each information set. Each information set uses RM or RM+ as its minimizer independently, and the update process is computed recursively from the bottom up in the game tree. The key difference lies in the computation of action loss in each information set, similar to Equation \eqref{v}, but referred to as counterfactual losses. These values assume that the player arrives at the information set with probability 1, while the opponent uses their strategy to reach that set. The expected value is computed in the subtree rooted at this information set. For an information set \(I \in \mathcal{I}_i\) for player \(i\), with action \(a \in A(I)\), the counterfactual loss is calculated as follows:

\begin{align}
    \ell^t_i(I, a) &= \sum_{h \in I a} \ell_i(h) \nonumber \\
    &= \sum_{h \in I a} \sum_{z \in Z: h \prec z} -q_{-i}^t(h) \sigma^t(z | h) u_i(z) \label{v_cfr_original}\\
    &= (-U_{i}q_{-i}^t)(Ia)+\sum_{I'\in C(I,a)}\sum_{a'\in A(I')}\sigma_i^t(I',a') \ell_i^t(I',a')
    \label{v_cfr}
\end{align}
where \(q_{-i}^t(h)\) represents the opponent's reach probability at state \(h\), and \(\sigma^t(z | h)\) represents the joint reach probability for the players to reach the terminal state \(z\) from state \(h\). The recursive formulation in Equation~\eqref{v_cfr}, computed in a bottom-up manner, equivalently represents the original CFR value function defined in Equation~\eqref{v_cfr_original}. The remaining steps in the dynamics are identical to those in RM/RM+, as given in Equations (\ref{r} to \ref{x}).

The convergence properties of CFR are based on Blackwell's approachability theorem, which guarantees the convergence of the average strategy—a weighted sum of the last strategies, with weights such as uniform (1), linear (\(t\)), or quadratic (\(t^2\)) as the iteration \(t\) progresses. All of these regret minimizers achieve an \(O(\sqrt{T})\) regret bound in the worst case, but often perform much better in practice than the theoretical bound suggests.

\subsection{Reward Transformation Framework}
The Reward Transformation (RT) framework augments the reward function with strongly convex regularization terms, referred to as the reward transformation term (RT-term). This transforms the minimax problem in Equation~\eqref{eq: minmax} into a strongly convex-concave optimization problem (SCCP) \citep{meng2023efficient}, formulated as:
\begin{equation}
    \min_{\sigma_1 \in \Sigma_1} \max_{\sigma_2 \in \Sigma_2} \langle \sigma_1, -U \sigma_2 \rangle + \mu \phi(\sigma_1, \sigma_1^r) - \mu \phi(\sigma_2, \sigma_2^r),
    \label{eq:sccp}
\end{equation}
where \(\phi\) is the regularization function, \(\mu > 0\) is the RT-term weight, and \(\sigma_1^r, \sigma_2^r\) are reference strategies for players 1 and 2, respectively. 
The inclusion of the RT-term modifies the reward gradient to:
\begin{equation}
    \ell_i = -U_i \sigma_{-i} + \mu \frac{\partial \phi(\sigma_i, \sigma_i^r)}{\partial \sigma_i},
    \label{eq:gradient}
\end{equation}
which depends on both the opponent's strategy \(\sigma_{-i}\) and the player's own strategy \(\sigma_i\). This creates an attractor for the dynamics \citep{perolat2021poincare}. For a fixed \(\mu\), the attractor corresponds to the saddle point \(\sigma^{*,r}\) of the SCCP in Equation~\eqref{eq:sccp}, constructed from the reference strategy \(\sigma^r\), which biases the NE \(\sigma^*\) of the original game.

\begin{figure}[ht]
    \centering
    \includegraphics[width=0.8\linewidth]{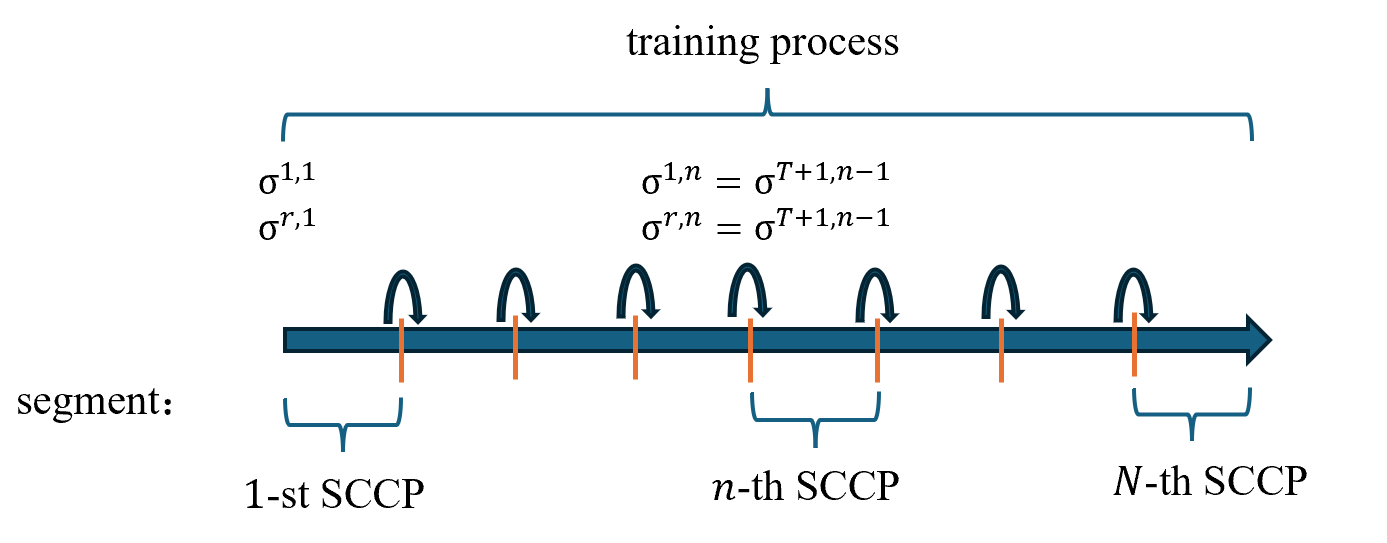}
    \caption{Reward Transformation Framework.}
    \label{fig:sccp}
\end{figure}

To learn the NE \(\sigma^*\) of the original game, the RT framework decomposes the iteration process into a sequence of SCCPs, as illustrated in Figure~\ref{fig:sccp}. In the \(n\)-th SCCP, the RT framework generates a sequence of strategies \(\{\sigma^{1,n}, \dots, \sigma^{T+1,n}\}\) over \(T\) iterations, where \(\sigma^{1,n}\) is an initial strategy inherited from the \((n-1)\)-th SCCP. The saddle point of the \(n\)-th SCCP, denoted \(\sigma^{*,n} := \sigma^{*, (r, n)}\), is defined with respect to the reference strategy \(\sigma^{r,n}\). The sequence of saddle points \(\{\sigma^{*,1}, \dots, \sigma^{*,n}\}\) converges to \(\sigma^*\), as guaranteed by \citet[Theorem 4.3]{meng2023efficient}, under the following assumption:
\begin{assumption}\label{as:sccp_converge}
Let \(\sigma^{T+1,n}\) be the strategy produced by the RT framework in the \(n\)-th SCCP after \(T\) iterations, and let \(\sigma^{*,n}\) be the saddle point of the \(n\)-th SCCP. We assume:
\[
\sigma^{T+1,n} = \sigma^{*,n}.
\]
\end{assumption}

By initializing the \(n\)-th SCCP with the final strategy of the \((n-1)\)-th SCCP as both the reference and initial strategy:
\begin{equation}
    \sigma^{r,n} = \sigma^{T+1,n-1}, \quad \sigma^{1,n} = \sigma^{T+1,n-1},
    \label{eq:initial_sccp}
\end{equation}
the RT framework achieves asymptotic convergence to the NE of the original game, i.e., \(\sigma^{t,n} \to \sigma^{*,n} \to \sigma^*\), though without a precise convergence rate.

\citet{meng2023efficient} propose RTRM+, which employs the Bregman distance with the Euclidean squared norm as the regularization function in Equation~\eqref{eq:sccp}. The dynamics modify only the action loss in Equation~\eqref{v} of RM+, as follows:
\begin{equation}
    \ell_i^t = -U \sigma_{-i}^t + \mu (\sigma_i^{t,n} - \sigma_i^{r,n}),
    \label{eq:rt_cfv}
\end{equation}
while other steps remain unchanged. RTRM+ achieves asymptotic last-iterate convergence within the RT framework \citep[Theorem 5.2]{meng2023efficient}.

\section{Convergence Analysis in the RT Framework}
In this section, we applied RT framework on RM-type algorithms by using Euclid square norm regularization, we analyze the convergence properties of the RT framework by establishing the relationship between the last strategy, the current saddle point of the SCCP, and the NE of the original game. This includes demonstrating the best-iterate convergence of the last strategy in RTRMs toward the saddle point of the SCCP and the asymptotic convergence to the original NE based on Assumption \ref{as:sccp_converge}. Finally, we address practical challenges where deviations from this assumption affect performance, motivating our adaptive method to enhance last-iterate convergence.

\subsection{Best-Iterate Convergence of RTRMs}
In order to satisfy the first convergence condition that satisfy Assumption \ref{as:sccp_converge}, we extend the results from \citet{meng2023efficient} and establish the best-iterate convergence of all RTRMs in each SCCP rather than only RTRM+ shows in \citet{meng2023efficient}:
\begin{theorem}[Best-iterate Convergence of RTRMs in the \(n\)-th SCCP]
    Given any reference strategy \(\sigma^{r,n}\) and RT-term weight \(\mu\), let \(\sigma^{t,n}\) be the strategy sequence produced by RTRMs and \(\sigma^{*,n} \in \Sigma^*\) be the saddle point of the \(n\)-th SCCP. Then, for any \(T \geq 1\), there exists \(0 < t \leq T\) such that
    
        \[
    \|\sigma^{*,n} - \sigma^{t,n}\|_2 \leq \frac{C}{\sqrt{T}}.
    \]
    where \(C>0\) is a constant depending on \(\mu,\sigma^{r,n}\) and the game.
    
    \label{thm:best_iterate_convergence_of_rtrms}
\end{theorem}

Directly using the RTRMs dynamics to prove that the last strategy converges to the saddle point is difficult. An alternative method is to convert it to the OMD formulation, leveraging the equivalence of these algorithms as shown in \citep{farina2021faster}. The proof details are provided in Appendix \ref{app:proof_of_thm:best}.

\subsection{Asymptotic Last-Iterate Convergence of RTRMs}
Based on Theorem \ref{thm:best_iterate_convergence_of_rtrms}, we ensure that the last strategy \(\sigma^{t,n}\) converges to the saddle point of the \(n\)-th SCCP. We now prove the second condition: that the sequence of saddle points \(\{\sigma^{*,1}, \dots, \sigma^{*,n}\}\) converges to the NE of the original game, which establishes asymptotic convergence of the last strategy.

We begin with the following lemma for the \(n\)-th SCCP:

\begin{lemma}
Let \(\sigma^{r,n}\) be a reference strategy, \(\mu > 0\) the RT-term weight, \(\sigma^{*,n} \in \Sigma^*\) the saddle point of the \(n\)-th SCCP, and \(\sigma^{*} \in \Sigma^*\) the NE of the original game. Provided \(\sigma^{r,n} \neq \sigma^{*,n} \neq \sigma^{*}\), the following hold:
\begin{equation}
    \|\sigma^{*} - \sigma^{*,n}\|_2 \leq \|\sigma^{*} - \sigma^{r,n}\|_2,
    \label{eq:saddle_reference_NE}
\end{equation}
and
\begin{equation}
    \|\sigma^{*} - \sigma^{r,n}\|_2 - \|\sigma^{*} - \sigma^{*,n}\|_2 \geq \frac{C^2}{\mu^2 (\|\sigma^{*} - \sigma^{r,n}\|_2 + \|\sigma^{*} - \sigma^{*,n}\|_2)},
    \label{eq:bound_reference_saddle}
\end{equation}
where \(C > 0\) is a constant depending only on the game.
\label{le:relation_of_3_points}
\end{lemma}

The proof is provided in Appendix~\ref{app:proof_of_le:relation_of_3_points}. We now establish the asymptotic convergence of RTRMs:

\begin{theorem}[Asymptotic Last-Iterate Convergence of RTRMs]
Let \(\{\sigma^{t,n}\}\) be the sequence generated by RTRMs. Then, \(\sigma^{t,n}\) is bounded and converges to the NE \(\sigma^{*} \in \Sigma^*\) of the original game.
\label{thm:asymptotic}
\end{theorem}

\begin{proof}
Let \(\sigma^{t,n}\) denote the last strategy of the \(n\)-th SCCP, \(\sigma^{*,n}\) its saddle point, and \(\sigma^{*}\) the NE of the original game. The distance from \(\sigma^{t,n}\) to \(\sigma^{*}\) is bounded by the triangle inequality:
\begin{equation}
    \|\sigma^{*} - \sigma^{t,n}\|_2 \leq \|\sigma^{*} - \sigma^{*,n}\|_2 + \|\sigma^{*,n} - \sigma^{t,n}\|_2.
    \label{eq:bound}
\end{equation}
Here, \(\|\sigma^{*} - \sigma^{*,n}\|_2\) measures the distance between the NE and the \(n\)-th saddle point, and \(\|\sigma^{*,n} - \sigma^{t,n}\|_2\) captures the inner iteration error.

By the RT framework initialization (Equation~\eqref{eq:initial_sccp}) and Assumption~\ref{as:sccp_converge}, we have \(\sigma^{r,n} = \sigma^{T+1,n-1} = \sigma^{*,n-1}\). Substituting into Equation~\eqref{eq:bound_reference_saddle}, we obtain:
\begin{equation}
    \|\sigma^{*} - \sigma^{*,n-1}\|_2 - \|\sigma^{*} - \sigma^{*,n}\|_2 \geq \frac{C^2}{\mu^2 (\|\sigma^{*} - \sigma^{*,n-1}\|_2 + \|\sigma^{*} - \sigma^{*,n}\|_2)}.
    \label{eq:difference_reference}
\end{equation}
Equation~\eqref{eq:difference_reference} implies that \(\{\|\sigma^{*} - \sigma^{*,n}\|_2\}\) is strictly decreasing, non-negative, and decreases faster when \(\sigma^{*,n-1}\) is closer to \(\sigma^{*}\). By the monotone convergence theorem, it converges to a limit \(L = 0\). Meanwhile, Theorem~\ref{thm:best_iterate_convergence_of_rtrms} ensures that \(\|\sigma^{*,n} - \sigma^{t,n}\|_2\) vanishes as \(t \to \infty\).

Thus,
\[
\lim_{n \to \infty, t \to \infty} \|\sigma^{*} - \sigma^{t,n}\|_2 = 0,
\]
establishing the asymptotic last-iterate convergence of RTRMs.
\end{proof}

\subsection{The Bottleneck of the RT Framework in Practical Settings}\label{sec:bottleneck}
\citet{meng2023efficient} identify a bottleneck in the RT framework, originally developed for the MWU algorithm, concerning the convergence rate of each SCCP. They extend the framework to the more practical RM+ algorithm. However, this bottleneck remains unaddressed, as the framework relies on the impractical Assumption~\ref{as:sccp_converge}. This assumption renders the theoretical convergence guarantees of the RT framework difficult to achieve in practice, as computing an exact saddle point in each SCCP iteration is generally infeasible. Typically, only an approximate saddle point is obtained, which violates Assumption~\ref{as:sccp_converge}. Consequently, the original RTRM+ method requires manual tuning of the interval \( T \) and the regularization parameter \( \mu \), resulting in suboptimal performance.

Theorem \ref{thm:best_iterate_convergence_of_rtrms} implies that the last strategy may oscillate around the saddle point. A fixed interval \(T\) may inadvertently select a suboptimal strategy to conclude the SCCP, which is then used as the reference in the next SCCP, thus seriously undermining convergence guarantees.

Moreover, the choice of the RT-term weight \(\mu\) in Equation \eqref{eq:sccp} significantly affects the convergence rate in practice, as established in Theorem~\ref{thm:best_iterate_convergence_of_rtrms}. Additionally, \(\mu\) bounds the distance between the saddle point of the SCCP and the NE of the original game~\citep{perolat2021poincare,abe2022mutation}. The detailed influence of \(\mu\) on convergence is discussed in Section~\ref{sec:adaptive_weight}. A fixed weight throughout the entire training process may lead to inefficiencies in convergence.

Furthermore, applying the RT framework to RTRMs brings unique challenges. The original dynamics of regret matching are specifically designed for average strategy convergence, where the strategy is updated using cumulative regret matching that accounts for the immediate regret from the first step to the current one with equal weighting. Since initial strategies are often far from the NE due to uniform initialization or being trained from scratch, using equal weighting across all iterations can be suboptimal. The RT framework, however, focuses on the last-iterate convergence, where the immediate regret tends to decrease as the last strategy converges. Thus, there is a question of whether RTRMs could achieve a more efficient cumulative regret process to align better with the goal of last-iterate convergence.

These observations motivated us to enhance the RT framework's performance by adaptively selecting the reference strategy, dynamically adjusting the weight parameter, and improving the efficiency of the regret accumulation method, as discussed in the next section.

\section{Adaptive RT Framework}
In this section, building on our earlier observations in Section~\ref{sec:bottleneck}, we present an adaptive mechanism comprising three adaptive methods. These methods include: (1) adaptively selecting the reference strategy to conclude the SCCP using an exploitability descent metric, (2) adjusting the weight \(\mu\) to balance convergence between the SCCP's saddle point and the NE of the original game, and (3) implementing a discounted cumulative regret technique that assigns greater weight to recent immediate regrets to enhance regret accumulation efficiency. These adaptive methods enable RTRMs to achieve linear last-iterate convergence throughout the entire training process.

\subsection{Adaptive Reference Strategy}
Lemma \ref{le:relation_of_3_points} implies that the saddle point of the SCCP is closer to the NE of the original game than the reference strategy. Thus, by appropriately handling the reference strategy, we can enhance convergence. To do so, we need a metric that measures the quality of the last strategy without knowing the NE of the original game. Exploitability is a useful metric for this purpose, as it represents the difference between a given strategy and the NE. Recall that zero-sum games satisfy the following metric subregularity condition:
    \begin{lemma}[Adapted from Saddle Point Metric Subregularity in \citealt{wei2020linear}]
    Let \(\sigma^* \in \Sigma^*\) be the Nash equilibrium of the game. There exists a constant \(C > 0\) (depending only on the game) such that for any strategy \(\sigma \in \Sigma \setminus \Sigma^*\), the following holds:
    \[
    \|\sigma^* - \sigma\|_2 \leq \frac{\epsilon(\sigma)}{C}.
    \]
    \label{lemma:MS}
    \end{lemma}

Applying Lemma \ref{le:relation_of_3_points} to Equation (\ref{eq:bound}), we have:
\begin{equation}
    \|\sigma^{*} - \sigma^{t,n}\|_2 < \|\sigma^{*} - \sigma^{r,n}\|_2 + \|\sigma^{*,n} - \sigma^{t,n}\|_2.
\label{eq:bound 2}
\end{equation}

Using Lemma \ref{lemma:MS}, we obtain:
\[
    \|\sigma^{*} - \sigma^{t,n}\|_2 < \frac{\epsilon(\sigma^{r,n})}{C} + O\left(\frac{1}{\sqrt{T}}\right).
\]

If we set \(\epsilon(\sigma^{r,n}) \leq \frac{\epsilon(\sigma^{r,n-1})}{2}\), then \(\|\sigma^{*} - \sigma^{t,n}\|_2 \leq O\left(\frac{1}{2^n}\right) + O\left(\frac{1}{\sqrt{T}}\right)\). This implies that we can compute the exploitability of the last strategy in each SCCP to select a suitable reference, enabling adaptive RTRMs to achieve linear last-iterate convergence in practice.

\subsection{Adaptive Reward Transformation Weight}\label{sec:adaptive_weight}
The RT-term transforms the reward, enabling convergence to the saddle point, but it also alters the position of the saddle point in the dynamics. The distance between the saddle point of an SCCP and the NE of the original game is influenced by the weight \(\mu\), as shown in Theorem \ref{theorem: mu bound}:

\begin{theorem}
    Let \(\sigma^{*} \in \Sigma^*\) and \(\sigma^{*,n} \in \Sigma^{*,n}\) be the Nash equilibrium of the original game and the saddle point of the \(n\)-th SCCP, respectively. Given any reference strategy \(\sigma^{r,n}\) and parameter \(\mu\), the exploitability of \(\sigma^{*,n}\) satisfies:
    \[
        \epsilon(\sigma^{*,n}) \leq 2\mu.
    \]
    \label{theorem: mu bound}
\end{theorem}

\begin{proof}
    By Equation~\eqref{eq:exp_saddle}, the exploitability is bounded as:
    \[
        \epsilon(\sigma^{*,n}) \leq \mu \|\sigma^* - \sigma^{*,n}\|_2 \cdot \|\sigma^{*,n} - \sigma^{r,n}\|_2.
    \]
    Since \(\|\sigma - \sigma'\|_2 \leq \sqrt{2}\) for all \(\sigma, \sigma' \in \Delta\), we have:
    \[
        \epsilon(\sigma^{*,n}) \leq \mu \cdot \sqrt{2} \cdot \sqrt{2} = 2\mu.
    \]
\end{proof}
This distance is proportional to \(\mu\). Moreover, \(\mu\) also controls the convergence rate of the last strategy toward the saddle point of the SCCP:

\begin{theorem}
    Given any reference strategy \(\sigma^{r,n}\) in the \(n\)-th SCCP, the convergence rate of the last strategy to the saddle point in RTRMs is proportional to \(\mu\).
    \label{theorem: mu converge rate}
\end{theorem}

We leave the proof in Appendix \ref{app: proof_of_thm:mu_convergence_rate}. Theorems \ref{theorem: mu bound} and \ref{theorem: mu converge rate} imply that a balanced approach for \(\mu\) is crucial in practice. This leads to the adaptation of increasing \(\mu\) when the reference strategy is a close approximation of the NE of the original game (as indicated by low exploitability) to accelerate convergence. Conversely, \(\mu\) should be decreased to encourage exploration toward the NE of the original game when the algorithm is trapped for a long time in local optimization within the SCCP. The adaptation is shown below by adjusting an adaptive weight for the RT term:

\begin{align}
    \ell_i^{t,n} &= -U \sigma^t_{-i} + w \cdot \mu (\sigma_i^{t,n} - \sigma_i^{r,n})
    \label{adaptive_cfv} \\
w &= \begin{cases}
        >1, & \text{if } \epsilon(\sigma^{r, n}) < \epsilon(\sigma^{r, n-1}) \\
        \leq 1, & \text{otherwise}
    \end{cases}
\end{align}

The choice of \(w\) should adapt to different games to achieve optimal performance. A moderate choice, such as \(w = 2\) for acceleration and \(w = 0.5\) for exploration, has proven effective in many settings.

\subsection{Adaptive Regret Weight}
\label{sec:adaptive_regret_weight}

To address the final limitations observed in Section~\ref{sec:bottleneck}, we propose an adaptive regret weight to enhance the convergence of cumulative regret in RM-like algorithms. Specifically, we identify that these algorithms lack a forgetfulness mechanism to accelerate last-iterate convergence. This issue, also noted in FTRL \citep{cai2024fast}, remains unaddressed, despite RM being equivalent to FTRL under certain conditions \citep{farina2021faster}. Inspired by the discounted weight approach in \citet{brown2019solving}, we introduce a method to efficiently mitigate this problem.

Consider a single-agent example from \citet{brown2019solving}, where an agent selects among three actions with utilities \(\ell = (1, 0, -10^6)\). The game has a pure NE at \(\sigma^* = (1, 0, 0)\). By Theorem 1 in \citet{cai2023last}, the original CFR+ dynamics (Equations~\eqref{v}--\eqref{x}) achieve last-iterate convergence, equivalent to RTRM+ without the RT-term. Starting with a uniform strategy \(\sigma^1 = (1/3, 1/3, 1/3)\), the first iteration yields regrets \(R^1 = r^1 = (333334, 333333, 0)\), updating the strategy to \(\sigma^2 \approx (0.5, 0.5, 0)\). Using uniform regret weights, it takes \(T = 471,407\) iterations for the cumulative regret to reach \(R^T = (4.71 \times 10^5, 6.81 \times 10^{-2}, 0)\), at which point the first action is chosen with near-unit probability. This slow convergence indicates that the cumulative regret remains trapped in suboptimal strategies (e.g., \(\sigma^1\)) for an extended period, despite immediate regrets favoring the first action. A forgetfulness mechanism is thus necessary to prioritize recent regret information.

The RM dynamics (Equations~\eqref{v}--\eqref{x}) are driven by cumulative regret. A natural approach is to introduce a non-decreasing weight sequence for the immediate regret sequence \(\{\hat{r}^1, \dots, \hat{r}^T\}\), where \(\hat{r}^t = w^t r^t\) and \(w^t \in (0, 1]\) satisfies \(w^i \leq w^j\) for all \(i \leq j\). The strategy update follows \(\sigma^{t+1} \propto \hat{R}^t = \sum_{k=1}^t w^k r^k\). Since \(r^t \propto 1 / \|\sigma^t - \sigma^*\|\), the immediate regret approaches zero as \(\sigma^t\) nears the saddle point.\footnote{At the saddle point, for a mixed strategy, all action regrets are zero; for a pure strategy, only the action with probability one has zero regret, while others have negative regrets. RM updates consider only the positive portion of the cumulative regret.} 

Additionally, RM dynamics (Equation~\eqref{x}) update strategies based solely on positive cumulative regret, ignoring negative regret. This is evident in RM+ outperforming RM, as it implicitly prioritizes positive regrets, though this may lead to worse regret bounds in some cases \citep{burch2019revisiting}. A refined weighting scheme, inspired by DCFR \citep{brown2019solving}, applies distinct weights to positive and negative cumulative regrets using parameters \(\alpha \geq \beta\):

\begin{equation}
\label{eq:dcfr}
\begin{aligned}
R_+ &= \frac{t^\alpha}{t^\alpha + 1} \cdot [R]_+, \\
R_- &= \frac{t^\beta}{t^\beta + 1} \cdot [R]_-.
\end{aligned}
\end{equation}

Thus, the weight is \(w^t = \prod_{k=t}^{T-1} \frac{k^z}{k^z + 1}\), where \(z \in \{\alpha, \beta\}\), ensuring non-decreasing weights. Discounted RM (DRM) generalizes RM algorithms by varying \(\alpha\) and \(\beta\). For instance, RM corresponds to \((\alpha, \beta) = (+\infty, +\infty)\), RM+ to \((\alpha, \beta) = (+\infty, -\infty)\), and Linear RM (LRM) to \((\alpha, \beta) = (1, 1)\).

As established in Theorem~\ref{thm:best_iterate_convergence_of_rtrms}, RTDRM guarantees convergence without improving the best-iterate rate. However, it serves as an effective forgetfulness mechanism to mitigate historical errors. For example, with \(\alpha = \beta = 1\), the single-agent example requires only \(T = 970\) iterations to select the first action, significantly faster than the original RM+. This alignment accelerates convergence toward the NE. Our experiments in Appendix~\ref{app:rtdcfr_test} evaluate parameter settings such as \((\alpha, \beta) \in \{1.5, 2\} \times \{-\infty, 0, 0.5\}\), confirming that appropriate parameters significantly improve performance.

\section{CFRs in the RT Framework for EFGs}
\label{sec:rtcfrs}
In this section, we explore the application of the RT framework to CFR for solving EFGs, termed RTCFRs. We demonstrate that RTCFRs exhibit last-iterate convergence, a desirable property for iterative solvers.

\subsection{Challenges with Dilated Regularization}
\label{subsec:dilated_difficulty}
In the work of \citet{meng2023efficient}, RTCFR+ defines the \(n\)-th SCCP using sequence-form strategies \(q \in \mathcal{Q}\) as:
\begin{equation}
    \min_{q_1 \in \mathcal{Q}_1} \max_{q_2 \in \mathcal{Q}_2} -q_1^\top U q_2 + \mu D_{\psi}(q_1, q_1^{r,n}) - \mu D_{\psi}(q_2, q_2^{r,n}),
    \label{eq:rtcfr_plus}
\end{equation}
where \(\mu > 0\) is a regularization parameter, and \(D_{\psi}\) is a Bregman divergence term based on the dilated Euclidean squared norm \citep{hoda2010smoothing,kroer2015faster,farina2025better}:
\[
    \psi(q_i) = \sum_{I \in \mathcal{I}_i} \beta_I \frac{q_i(pI)}{2} \sum_{a \in A(I)} \frac{q_i(Ia)}{q_i(pI)},
\]
with \(\beta_I\) as a dilated weight for information set \(I \in \mathcal{I}_i\). The counterfactual loss is defined as:
\begin{equation}
    \hat{\ell}_i^t(I,a) = \ell_i^t(I,a) + \sum_{I' \in C(I,a)} \frac{q_i^t(Ia)}{q_i^t(pI)} \left\langle \hat{\ell}_i^t(I'), \frac{q_i^t(I')}{q_i^t(Ia)} \right\rangle,
    \label{eq:counterfactual_loss}
\end{equation}
where \(\ell_i^t(I,a)\) is given by:
\begin{align}
    \ell_i^t(I,a) &= \langle U_i, q_{-i}^t \rangle(Ia) + \mu \beta_I \left( \frac{q_i(Ia)}{q_i(pI)} - \frac{q_i^r(Ia)}{q_i^r(pI)} \right) \notag\\
    &\quad- \mu \sum_{I' \in C(I,a)} \sum_{a' \in A(I')} \frac{\beta_{I'}}{2} \left( \frac{(q_i(I'a'))^2}{(q_i(Ia))^2} - \frac{(q_i^r(I'a'))^2}{(q_i^r(Ia))^2} \right).
    \label{eq:counterfactual_loss_term}
\end{align}
    
The term involving the double summation in Equation~\eqref{eq:counterfactual_loss_term} introduces significant complexity, hindering the proof of convergence for RTCFR+ \citep[Appendix F]{meng2023efficient}. Consequently, \citet{meng2023efficient} omits this term for implementation simplicity, without providing a convergence guarantee. We observe that dilated regularization is less suitable for CFR in the counterfactual framework, as it updates strategies at ancestor information sets based on optimized strategies and utilities of descendant sets, which is not required by standard CFR.

\subsection{RTCFRs with Laminar Regret Decomposition}
\label{subsec:laminar_decomposition}
To address these challenges in Section~\ref{subsec:dilated_difficulty}, we propose RTCFRs using the laminar regret decomposition introduced by \citet{farina2019online}. RTCFRs optimize each information set independently using RTRMs as the solver. The update scheme proceeds in a bottom-up fashion, incorporating the RT-term into the counterfactual regret computation. For a subtree \(\triangle_I\) rooted at an information set \(I \in \mathcal{I}_i\) for player \(i\), the subtree strategy \(\sigma_i(\triangle_I) \in \Sigma_i(\triangle_I) \subseteq \Sigma_i\) forms a convex and compact set. The subtree value, incorporating the RT-term, is defined as:
\begin{equation}
    \hat{V}_{\triangle_I}^t(\sigma_i^t(\triangle_I)) = V_{\triangle_I}^t(\sigma_i^t(\triangle_I)) + \sum_{I' \in \triangle_I} \hat{\pi}(I \to I') \mu D_{\psi}(\sigma_i^t(I'), \sigma_i^r(I')),
    \label{eq:subtree_value_rt}
\end{equation}
where \(\hat{\pi}(I \to I') = \frac{\hat{q}_i(pI')}{\hat{q}_i(pI)}\) denotes the reach probability from the subtree root \(I\) to an information set \(I' \in \triangle_I\) under the saddle-point strategy, which does not require computation in the counterfactual framework. The subtree value \(V_{\triangle_I}^t\) is given by:
\begin{equation}
    V_{\triangle_I}^t(\sigma_i^t(\triangle_I)) = \langle g_i^t(I), \sigma_i^t(I) \rangle + \sum_{a \in A(I)} \sum_{I' \in C(I,a)} \sigma_i^t(I,a) V_{\triangle_{I'}}^t(\sigma_i^t(\triangle_{I'})),
    \label{eq:subtree_value}
\end{equation}
where \(g_i^t(I) = \langle -U_i(I), q_{-i}^t \rangle\) represents the observed loss at information set \(I\), with \(U_i(I) \in \mathbb{R}^{|A(I)| \times |\mathcal{S}_{-i}|}\) as the utility matrix restricted to \(I\). Note that \(V_{\triangle_{I'}}^t\) is treated as a constant in Equation~\eqref{eq:subtree_value}, as updates in the counterfactual framework are independent across information sets, and \(\sigma_i^t(\triangle_{I'})\) is fixed in the current iteration. This contrasts with dilated methods, which depend on the updated subtree strategy \(\sigma_i^{t+1}(\triangle_{I'})\), despite both using bottom-up updates.

The value at each information set \(I\) depends only on the local strategy \(\sigma_i(I) \in \Delta^{|A(I)|}\):
\begin{equation}
    \hat{V}_{I}^t(\sigma_i(I)) = \langle g_i^t(I), \sigma_i(I) \rangle + \mu D_{\psi}(\sigma_i(I), \sigma_i^r(I)) + \sum_{a \in A(I)} \sum_{I' \in C(I,a)} \sigma_i(I,a) V_{\triangle_{I'}}^t(\sigma_i^t(\triangle_{I'})).
    \label{eq:laminar_value}
\end{equation}
The counterfactual loss corresponds to the gradient of Equation~\eqref{eq:laminar_value}, as implemented in lines 5, 7 and 16 of Algorithm~\ref{alg:adaptive-rtcfr}:
\begin{equation}
    \ell_i^t(I) = g_i^t(I) + \mu (\sigma_i(I) - \sigma_i^r(I)) + \Bigg( \sum_{I' \in C(I,a)} V_{\triangle_{I'}}^t(\sigma_i^t(\triangle_{I'}))\Bigg)_{a\in A(I)}.
    \label{eq:counterfactual_loss_gradient}
\end{equation}

The regret with the RT term at subtree \(\triangle_I\) is:
\begin{align}
    \hat{r}_{\triangle_I}^t 
    &= \hat{V}_{\triangle_I}^t(\sigma_i^t(\triangle_I)) - \min_{\hat{\sigma}_i(\triangle_I)} \hat{V}_{\triangle_I}^t(\hat{\sigma}_i(\triangle_I))  \label{eq:min_subtree_value}\\
    &\stackrel{\eqref{eq:subtree_value_rt}}{=} V_{\triangle_I}^t(\sigma_i^t(\triangle_I)) -\min_{\hat{\sigma}_i(\triangle_I)}\Bigg\{ \langle g_i^t(I),\hat{\sigma}_i(I)\rangle+\sum_{a\in A(I)}\sum_{I'\in C(I,a)}\hat{\sigma}_i(I,a)V_{\triangle_{I'}}^t(\hat{\sigma}_i(\triangle_{I'})) \notag\\
    &\quad + \sum_{\dot{I}\in \triangle_I}\frac{\hat{q}_i(p\dot{I})}{\hat{q}_i(pI)}\mu (D_{\psi}(\hat{\sigma}_i(\dot{I}),\sigma_i^r(\dot{I}))-D_{\psi}(\sigma_i^t(\dot{I}),\sigma_i^r(\dot{I})))\Bigg\} \notag\\
    &= V_{\triangle_I}^t(\sigma_i^t(\triangle_I)) + \mu D_{\psi}(\sigma_i^t(I),\sigma_i^r(I))-\min_{\hat{\sigma}_i(I)\in \Delta^{|A(I)|}}\Bigg\{ \langle g_i^t(I),\hat{\sigma}_i(I)\rangle+\mu D_{\psi}(\hat{\sigma}_i(I),\sigma_i^r(I)) \notag\\
    &\quad+\sum_{a \in A(I)} \sum_{I' \in C(I,a)} \hat{\sigma}_i(I,a) \min_{\hat{\sigma}_i(\triangle_{I'})} \bigg(V_{\triangle_{I'}}^t(\hat{\sigma}_i(\triangle_{I'})) \notag\\
    &\quad + \sum_{\dot{I}\in \triangle_{I'}}\frac{\hat{q}_i(p\dot{I})}{\hat{q}_i(pI')}\mu (D_{\psi}(\hat{\sigma}_i(\dot{I}),\sigma_i^r(\dot{I}))-D_{\psi}(\sigma_i^t(\dot{I}),\sigma_i^r(\dot{I})))\bigg) \Bigg\} \label{eq:31} \\
    &\stackrel{(\ref{eq:subtree_value_rt}, \ref{eq:laminar_value})}{=}\hat{V}_{I}^t(\sigma_i^t(I)) - \min_{\hat{\sigma}_i(I) \in \Delta^{|A(I)|}} \left\{ \hat{V}_{I}^t(\hat{\sigma}_i(I)) - \sum_{a \in A(I)} \sum_{I' \in C(I,a)} \hat{\sigma}_i(I,a) \min_{\hat{\sigma}_i(\triangle_{I'})}\hat{V}_{\triangle_{I'}}^t(\hat{\sigma}_i(\triangle_{I'})) \right\} \notag\\
    &\stackrel{\eqref{eq:min_subtree_value}}{=} \hat{V}_{I}^t(\sigma_i^t(I)) - \min_{\hat{\sigma}_i(I) \in \Delta^{|A(I)|}} \left\{ \hat{V}_{I}^t(\hat{\sigma}_i(I)) - \sum_{a \in A(I)} \sum_{I' \in C(I,a)} \hat{\sigma}_i(I,a) \hat{r}_{\triangle_{I'}}^t \right\},
    \label{eq:laminar_regret}
\end{align}
Equation~\eqref{eq:laminar_regret} implies that regret minimization in RTCFRs can be performed independently at each information set \(I\):
\begin{equation}
    \hat{r}^t(I) = \hat{V}_{I}^t(\sigma_i^t(I)) - \min_{\hat{\sigma}_i(I) \in \Delta^{|A(I)|}} \hat{V}_{I}^t(\hat{\sigma}_i(I)).
    \label{eq:info_set_regret}
\end{equation}
The overall regret is bounded by a weighted sum of per-information-set regrets:
\begin{lemma}[Adapted from Theorem 2 in \citet{farina2019online}]
    \label{lem:regret_bound}
    The regret on \(\Sigma\) satisfies:
    \[
        r^t \leq \max_{\hat{\sigma} \in \Sigma} \sum_{I \in \mathcal{I}} \hat{q}_i(pI) \hat{r}^t(I),
    \]
    where \(\hat{q}_i(pI)\) is the sequence-form strategy for \(\hat{\sigma}_i\), representing the reach probability to information set \(I\) for player \(i = P(I)\).
\end{lemma}
Thus, RTCFRs solve the \(n\)-th SCCP for each player \(i \in \{1, 2\}\) in EFGs as:
\begin{equation}
    q_i^{*,n} = \argmin_{q_i \in \mathcal{Q}_i} \sum_{I \in \mathcal{I}_i} q_i(pI) \hat{V}_I(\sigma_i(I)),
    \label{eq:sccp_efg}
\end{equation}
where \(\sigma_i(I) = \frac{q_i(Ia)}{q_i(pI)}\), and \(q_i(pI) \hat{V}_I(\sigma_i(I)) = 0\) if \(q_i(pI) = 0\).

We further establish that RTCFRs, leveraging RTRMs within information sets, achieve best-iterate convergence for the \(n\)-th SCCP in sequence-form strategies:
\begin{theorem}[Best-Iterate Convergence of RTCFRs in the \(n\)-th SCCP]
    \label{thm:best_iterate_convergence_of_RTCFRs}
    Given a reference strategy \(q^{r,n}\) and RT-term weight \(\mu\), let \(\{q^{t,n}\}\) be the sequence of strategies produced by RTCFRs, and let \(q^{*,n} \in \mathcal{Q}^{*,n}\) be the saddle point of the \(n\)-th SCCP. For any \(T \geq 1\), there exists \(0 < t \leq T\) such that:
    \[
        \|q^{*,n} - q^{t,n}\| = \sum_{I \in \mathcal{I}} q^{*,n}(pI) \|\sigma^{*,n}(I) - \sigma^{t,n}(I)\|_2 \leq O\left(\frac{1}{\sqrt{T}}\right).
    \]
\end{theorem}
The proof is provided in Appendix~\ref{app:proof_of_thm:best_iterate_convergence_of_RTCFRs}. The asymptotic last-iterate convergence of RTCFRs aligns with the results in Theorem~\ref{thm:asymptotic}.

\begin{theorem}[Asymptotic Last-Iterate Convergence of RTCFRs]
    \label{thm:asymptotic_last_iterate_converge_of_RTCFRs}
    Let \(G = (\mathcal{Q}_1, \mathcal{Q}_2, U)\) be an EFG, and let \(\{q^{*,1}, \dots, q^{*,n}\}\) be the saddle-point sequence across \(n\) SCCPs. Then, \(q^{*,n}\) is bounded, converges to the Nash Equilibrium (NE) \(q^* \in \mathcal{Q}^*\) of \(G\), and the strategy \(q^{t,n}\) produced by RTCFRs asymptotically converges to \(q^*\).
\end{theorem}
\begin{proof}
    Let \(q^* \in \mathcal{Q}^*\) be the NE of \(G\), and let \(q^{*,n} \in \mathcal{Q}^{*,n}\) and \(q^{r,n}\) be the saddle point and reference strategy of the \(n\)-th SCCP, respectively. By the saddle-point property of Equation~\eqref{eq:sccp_efg}, for each information set \(I \in \mathcal{I}_i\):
    \begin{equation}
        q_i^{*,n}(pI) \langle \sigma_i^*(I), \ell_i^{*,(*,n)}(I) \rangle \geq q_i^{*,n}(pI) \langle \sigma_i^{*,n}(I), \ell_i^{*,n}(I) \rangle,
        \label{eq:laminar_saddle_property}
    \end{equation}
    where:
    \[
        \ell_i^{*,(*,n)}(I) = g_i^{*,n}(I) + \mu (\sigma_i^{*,n}(I) - \sigma_i^{r,n}(I)) + \Bigg(\sum_{I' \in C(I,a)} V_{\triangle_{I'}}^{*,n}(\sigma_i^*(\triangle_{I'}))\Bigg)_{a\in A(I)},
    \]
    and \(\ell_i^{*,n}(I) = g_i^{*,n}(I) + \mu (\sigma_i^{*,n}(I) - \sigma_i^{r,n}(I)) + \bigg( \sum_{I' \in C(I,a)} V_{\triangle_{I'}}^{*,n}(\sigma_i^{*,n}(\triangle_{I'}))\bigg)_{a\in A(I)}\), with \(g_i^{*,n}(I) = \langle -U_i(I), q_{-i}^{*,n} \rangle\) and \(V_{\triangle_I}^{*,n}(\sigma_i)=\langle\sigma_i,g_i^{*,n}\rangle+\sum_{a \in A(I)} \sum_{I' \in C(I,a)} \sigma_i(I,a) V_{\triangle_{I'}}^{*,n}(\sigma_i(\triangle_{I'}))\). Extending Equation~\eqref{eq:laminar_saddle_property}, we obtain:
    \begin{align}
        \langle q_i^*(\triangle_I), -U_i(\triangle_I) q_{-i}^{*,n} \rangle + q_i^{*,n}(pI)\langle \sigma_i^*(I), \mu (\sigma_i^{*,n}(I) &- \sigma_i^{r,n}(I)) \rangle \geq \langle q_i^{*,n}(\triangle_I), -U_i(\triangle_I) q_{-i}^{*,n} \rangle \notag\\
        &+ q_i^{*,n}(pI) \langle \sigma_i^{*,n}(I), \mu (\sigma_i^{*,n}(I) - \sigma_i^{r,n}(I)) \rangle.
        \label{eq:extended_saddle}
    \end{align}
    Summing over subtrees rooted at information sets \(I\) for players \(i \in \{1, 2\}\), we have:
    \[
        \sum_{i \in \{1, 2\}} \langle q_i^*(\triangle_I), -U_i(\triangle_I) q_{-i}^{*,n} \rangle \leq \sum_{i \in \{1, 2\}} \langle q_i^{*,n}(\triangle_I), -U_i(\triangle_I) q_{-i}^{*,n} \rangle,
    \]
    yielding, for each information set:
    \[
        q_i^{*,n}(pI) \langle \sigma_i^*(I), \mu (\sigma_i^{*,n}(I) - \sigma_i^{r,n}(I)) \rangle \geq q_i^{*,n}(pI) \langle \sigma_i^{*,n}(I), \mu (\sigma_i^{*,n}(I) - \sigma_i^{r,n}(I)) \rangle.
    \]
    Following the proof of Lemma~\ref{le:relation_of_3_points} in Appendix~\ref{app:proof_of_le:relation_of_3_points}, we derive:
    \[
    q_i^{*,n}(pI)(\|\sigma_i^{*}(I) - \sigma_i^{r,n}(I)\|_2 - \|\sigma_i^{*}(I) - \sigma_i^{*,n}(I)\|_2) \geq \frac{q_i^{*,n}(pI)C^2}{\mu^2 (\|\sigma_i^{*}(I) - \sigma_i^{r,n}(I)\|_2 + \|\sigma_i^{*}(I) - \sigma_i^{*,n}(I)\|_2)},
    \]
    define \(C_{\text{min}}=\min_{I\in \mathcal{I}}\frac{C^2}{\mu^2 (\|\sigma^{*}(I) - \sigma^{r,n}(I)\|_2 + \|\sigma^{*}(I) - \sigma^{*,n}(I)\|_2)}\), and
    \begin{align}
        \|q^{*}(\triangle_I) - q^{r,n}(\triangle_I)\|-\|q^{*}(\triangle_I) - q^{*,n}(\triangle_I)\| &= q^{*,n}(pI) \bigg( \|\sigma^{*}(I) - \sigma^{r,n}(I)\|_2 - \|\sigma^{*}(I) - \sigma^{*,n}(I)\|_2 \notag\\ 
        & \hspace*{-13em}  
         + \sum_{a \in A(I)} \sum_{I' \in C(I,a)} \frac{q^{*,n}(Ia)}{q^{*,n}(pI)} (\| (q^{*}(\triangle_{I'}) - q^{r,n}(\triangle_{I'})\|-\|q^{*}(\triangle_{I'}) - q^{*,n}(\triangle_{I'})\|) \bigg)\notag\\
        &\geq M_{Q_{\triangle_I}}C_{\text{min}}. \notag
    \end{align}
    where \(M_{Q_{\triangle_I}}=\max_{q\in \mathcal{Q}_{\triangle_I}} {\|q\|_1}\). For the entire strategy space \(q\in \mathcal{Q}\), we have:
    \begin{equation}
        \|q^{*} - q^{r,n}\|-\|q^{*}-q^{*,n}\| \geq M_{Q}C_{\text{min}}
    \end{equation}
    
    By the RT framework initialization and Assumption~\ref{as:sccp_converge}, substituting \(q^{r,n} = q^{*,n-1}\), the sequence \(\{\|q^{*}-q^{*,n}\|\}\) is strictly decreasing and non-negative, converging to zero. Thus, \(q^{*,n}\) converges to \(q^*\). Combined with Theorem~\ref{thm:best_iterate_convergence_of_RTCFRs}, which establishes that \(q^{t,n}\) converges to \(q^{*,n}\), we conclude that RTCFRs' strategies \(q^{t,n}\) asymptotically converge to the NE \(q^*\).
\end{proof}

\begin{algorithm}[tbh]
\caption{Adaptive RTCFRs Algorithm for Extensive-Form Games}
\begin{algorithmic}[1]
\Require Initial weight \(\mu > 0\), maximum iterations per SCCP \(T \in \mathbb{N}\)
\State Initialize local regret minimizers \(\mathcal{R}_I\): \(\sigma^1(I) \gets \frac{\mathbf{1}}{|A(I)|}\), \(R^0(I) \gets \mathbf{0}\), \(\forall I \in \mathcal{I}\)
\State Initialize reference strategy \(\sigma^r \gets \sigma^1\), minimum exploitability \(\epsilon_{\min} \gets \epsilon(\sigma^1)\), adaptive weight \(w \gets 1\), counter \(k \gets 0\)
\For{\(t = 1, 2, \ldots\)}
    \For{each player \(i \in \{1, 2\}\)}
        \State Compute counterfactual loss: \(\ell_i^t \gets \langle -\mathbf{U}, q_{-i}^t \rangle\) 
        \For{each information set \(I \in \mathcal{I}_i\) (bottom-up)}
            \State Compute counterfactual loss with RT-term: \(\hat{\ell}_i^t(I)\gets \ell_i^t(I)+w \mu (\sigma_i^t(I) - \sigma_i^r(I))\)
            \State Compute immediate regret: \(r_i^t(I) \gets \langle \sigma_i^t(I), \hat{\ell}_i^t(I) \rangle \mathbf{1} - \hat{\ell}_i^t(I)\)
            \State Update cumulative regret:
                \[
                R_i^t(I) \gets \begin{cases}
                    R_i^{t-1}(I) + r_i^t(I), & \text{if RM} \\
                    [R_i^{t-1}(I) + r_i^t(I)]_+, & \text{if RM+} \\
                    \frac{t^{\alpha}}{t^{\alpha} + 1} [R_i^{t-1}(I)+ r_i^t(I)]_+ + \frac{t^{\beta}}{t^{\beta} + 1} [R_i^{t-1}(I)+ r_i^t(I)]_-, & \text{if DRM}
                \end{cases}
                \]
            \State Compute positive regret: \(\theta^t(I) \gets [R_i^t(I)]_+\)
            \If{\(\theta^t(I) \neq \mathbf{0}\)}
                \State Update strategy: \(\sigma_i^{t+1}(I) \gets \frac{\theta^t(I)}{\|\theta^t(I)\|_1}\)
            \Else
                \State Set uniform strategy: \(\sigma_i^{t+1}(I) \gets \frac{\mathbf{1}}{|A(I)|}\)
            \EndIf
            \State Update parent counterfactual loss: \(\ell_i^t(pI) \gets \ell_i^t(pI) + \langle \sigma_i^t(I), \ell_i^t(I) \rangle\)
        \EndFor
        \State Increment counter: \(k \gets k + 1\)
    \EndFor
    \If{\(t \mod m = 0\)} \Comment{Check exploitability every \(m\) iterations}
        \If{\(\epsilon(\sigma^t) \leq \frac{\epsilon_{\min}}{2}\)} \Comment{Exploit phase}
            \State Update minimum exploitability: \(\epsilon_{\min} \gets \epsilon(\sigma^t)\)
            \State Update for next SCCP: \(\sigma^r \gets \sigma^t\), \(w \gets 2\), \(k \gets 0\)
        \ElsIf{\(\epsilon(\sigma^t) \leq \epsilon_{\min} \wedge k \geq T\)} \Comment{Keep phase}
            \State Update minimum exploitability: \(\epsilon_{\min} \gets \epsilon(\sigma^t)\)
            \State Update for next SCCP: \(\sigma^r \gets \sigma^t\), \(w \gets 1\), \(k \gets 0\)
        \ElsIf{\(k \geq 2T\)} \Comment{Explore phase}
            \State Update for next SCCP: \(\sigma^r \gets \sigma^t\), \(w \gets 0.5\), \(k \gets 0\)
        \EndIf
    \EndIf
\EndFor
\State \Return \(\sigma^{t+1}\)
\end{algorithmic}
\label{alg:adaptive-rtcfr}
\end{algorithm}
We introduce the Adaptive RTCFRs algorithm for EFGs in Algorithm~\ref{alg:adaptive-rtcfr}. The Adaptive RTRMs algorithm for NFGs is a special case of RTCFRs, where each player has a single information set. The adaptive mechanism, which enhances RTCFRs with linear last-iterate convergence, is detailed in lines 20--29 and operates in three phases:
\begin{itemize}
    \item \textbf{Exploit Phase} (lines 21--23): If the exploitability of the current strategy \(\sigma^t\) is at most half the minimum exploitability (\(\epsilon(\sigma^t) \leq \epsilon_{\min}/2\)), the algorithm adopts \(\sigma^t\) as the reference strategy (\(\sigma^r \gets \sigma^t\)) and transitions to a new SCCP with an aggressive adaptive weight \(w \gets 2\).
    \item \textbf{Keep Phase} (lines 24--26): If the exploitability satisfies \(\epsilon(\sigma^t) \leq \epsilon_{\min}\) and at least \(T\) iterations have elapsed in the current SCCP (\(k \geq T\)), the algorithm updates the reference strategy (\(\sigma^r \gets \sigma^t\)) and transitions to a new SCCP with a reset weight \(w \gets 1\).
    \item \textbf{Explore Phase} (lines 27--28): If no suitable reference strategy is identified after \(2T\) iterations (\(k \geq 2T\)), the algorithm enters an exploration phase with a conservative weight \(w \gets 0.5\).
\end{itemize}
Consequently, each SCCP comprises 1 to \(2T\) iterations. To mitigate the computational cost of exploitability calculations in large games, exploitability is evaluated every \(m\) iterations (\(m < T\), default \(m=1\)), as specified in line 20. This optimization preserves performance while improving computational efficiency.

\section{Experiments}\label{sec:experiments}
In our experiments, we evaluated our technique on NFG and EFG benchmarks, including matrix games, Kuhn poker, Leduc poker, Liar's dice, and Goofspiel. the performance is quantified by exploitability. All experiments were conducted on a PC with a 24-core CPU (up to 5.8 GHz) and 32 GB of memory, with most tests completing in under a few minutes. Detailed descriptions of the experiments can be found in the Appendix \ref{appendix: experiments}.

\subsection{Experiments on NFGs}
\label{sec:experiments_nfg}

In NFGs, we evaluated the performance of algorithms on a variety of randomly generated matrix games, where each player has 10 actions, and utility matrices were generated in the range \([-1, 1]\) using random seeds. We applied our adaptive techniques to RM-type algorithms, specifically Adp-RTRM+ (Adaptive RTRM+) and Adp-RTDRM (Adaptive RTDRM), and compared them with last-iterate convergence algorithms: RTRM+ \citep{meng2023efficient}, RTDRM, PRM+ (last) \citep{cai2023last}, OMWU \citep{lee2021last}, Reg-OMWU (Regularized OMWU) \citep{liu2022power}, and R-NaD (MWU with an RT term as described in \citealp{perolat2021poincare}).Additionally, we compared these methods with several average-iterate convergence algorithms, including RM+ and PRM+ (Predictive RM+) \citep{farina2021faster}, all of which utilized quadratic weighting to average the historical strategy. MWU-type algorithms employed entropy regularization. We set a uniform strategy as the initial strategy for all algorithms and employed an alternating update scheme, as it generally yielded better empirical performance. Further details on parameter selection and tuning can be found in Appendix \ref{appendix: exp setting}.

Figure \ref{fig:matrix_game} presents the results, demonstrating that our adaptive methods significantly enhance the convergence of the RTRM algorithms, and that discounted RM variants (DRM, RTDRM, and adaptive RTDRM) outperform RM+ types. Moreover, the effectiveness of the RT framework for last-iterate convergence is evident, as it outperforms other algorithms, including PRM and OMWU. Overall, the last-iterate convergence algorithms demonstrate faster convergence compared to averaging methods in this NFG setting, achieving a bound as low as \(10^{-16}\) and converging to a highly precise NE. 
\begin{figure}[ht]
    \centering
    \begin{subfigure}[b]{0.48\linewidth}
        \centering
        \includegraphics[width=\linewidth]{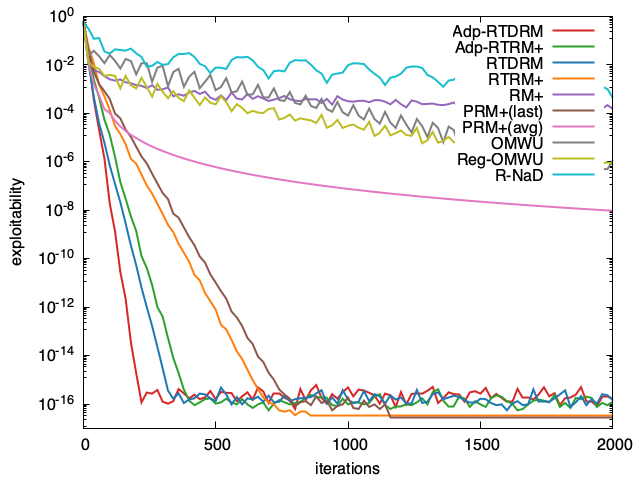}
        \caption{Matrix Game (Seed 0)}
        \label{fig:m-0}
    \end{subfigure}
    \hfill 
    \begin{subfigure}[b]{0.48\linewidth}
        \centering
        \includegraphics[width=\linewidth]{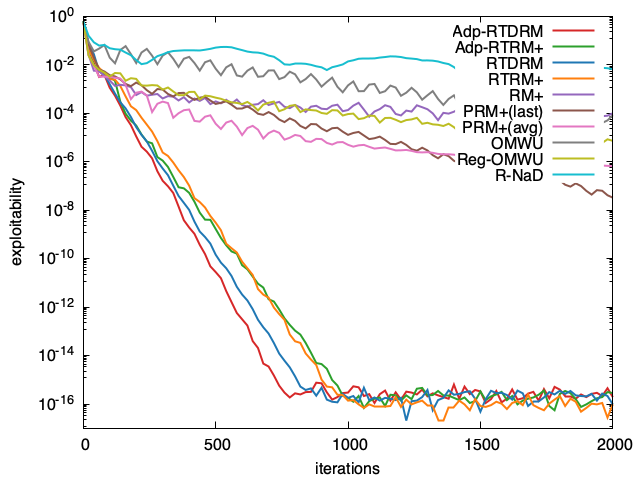}
        \caption{Matrix Game (Seed 1)}
        \label{fig:m-1}
    \end{subfigure}

    \begin{subfigure}[b]{0.48\linewidth}
        \centering
        \includegraphics[width=\linewidth]{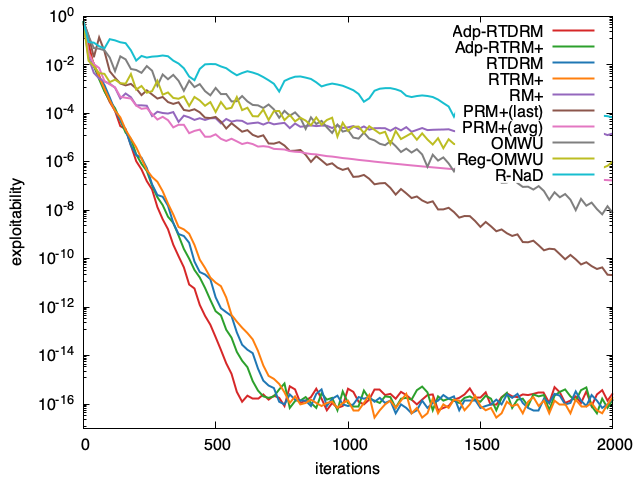}
        \caption{Matrix Game (Seed 2)}
        \label{fig:m-2}
    \end{subfigure}
    \hfill 
    \begin{subfigure}[b]{0.48\linewidth}
        \centering
        \includegraphics[width=\linewidth]{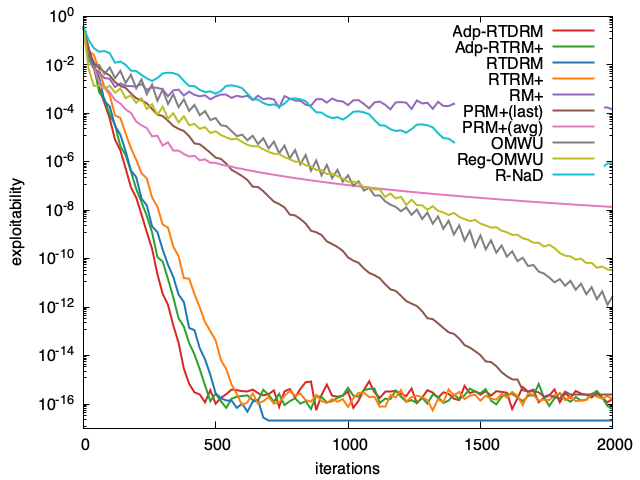}
        \caption{Matrix Game (Seed 3)}
        \label{fig:m-3}
    \end{subfigure}

    \caption{Results of the \(10 \times 10\) Matrix Game for seeds ranging from 0 to 3.}
    \label{fig:matrix_game}
\end{figure}

\subsection{Experiments on EFGs}
\label{sec:experiments_efg}
In EFGs, we applied our adaptive technique to CFR+ and DCFR, naming them Adp-RTCFR+ (Adaptive RTCFR+) and Adp-RTDCFR (Adaptive RTDCFR). We evaluated these methods on four EFG benchmark platforms: Kuhn poker, Leduc poker, Liar's dice, and Goofspiel. Detailed descriptions of these games are provided in Appendix \ref{appendix: game description}. We compared our methods with last-iterate convergence algorithms, including RTCFR+ \citep{meng2023efficient}, RTDCFR (RT Discounted CFR), PCFR+ (last) \citep{cai2023last}, DOMWU (Dilated OMWU, which employs dilated regularization weights in EFG \citealp{hoda2010smoothing, kroer2015faster, farina2021faster}), Reg-DOMWU (Regularized DOMWU) \citep{liu2022power}, and R-NaD \citep{perolat2021poincare}, as well as average-iterate convergence algorithms, including CFR+ \citep{bowling2015heads} and PCFR+ (Predictive CFR+) \citep{farina2021faster}. Similar to the NFG setting, we set a uniform strategy as the initial strategy for all algorithms and utilized an alternating update scheme. Details on parameter selection and tuning are available in Appendix \ref{appendix: exp setting}.
\begin{figure}[ht]
    \centering
    \begin{subfigure}[b]{0.48\linewidth}
        \centering
        \includegraphics[width=\linewidth]{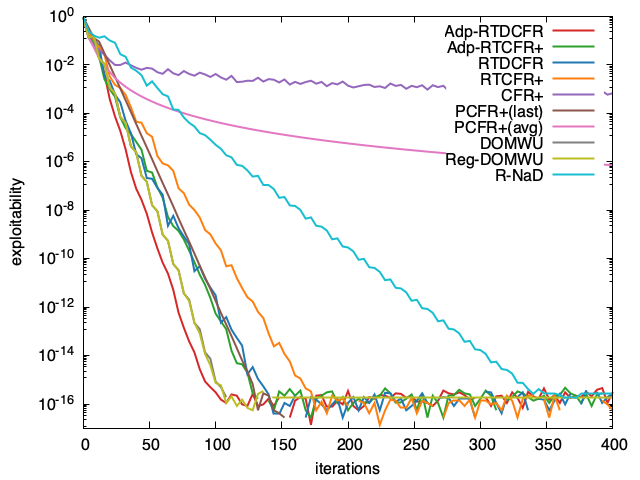}
        \caption{Kuhn Poker (3)}
        \label{fig:kuhn}
    \end{subfigure}
    \hfill 
    \begin{subfigure}[b]{0.48\linewidth}
        \centering
        \includegraphics[width=\linewidth]{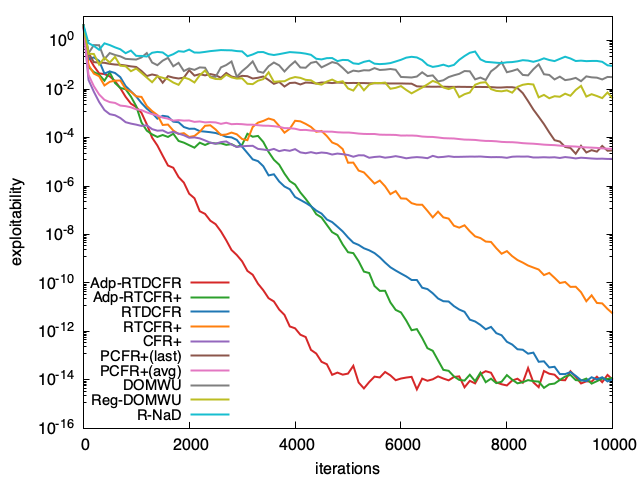}
        \caption{Leduc Poker (3)}
        \label{fig:leduc}
    \end{subfigure}

    \begin{subfigure}[b]{0.48\linewidth}
        \centering
        \includegraphics[width=\linewidth]{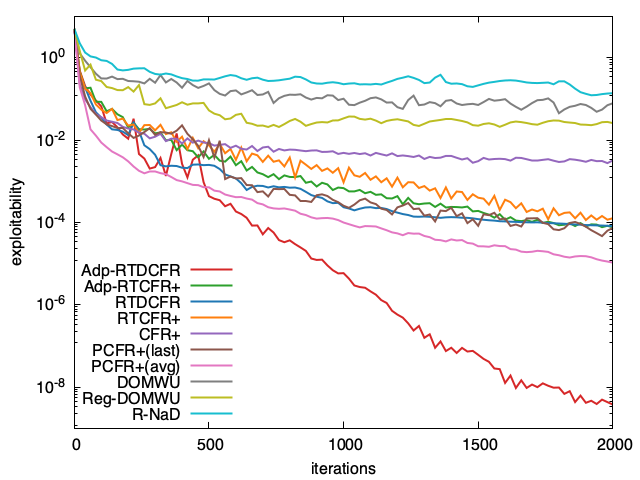}
        \caption{Goofspiel (4)}
        \label{fig:goofspiel}
    \end{subfigure}
    \hfill 
    \begin{subfigure}[b]{0.48\linewidth}
        \centering
        \includegraphics[width=\linewidth]{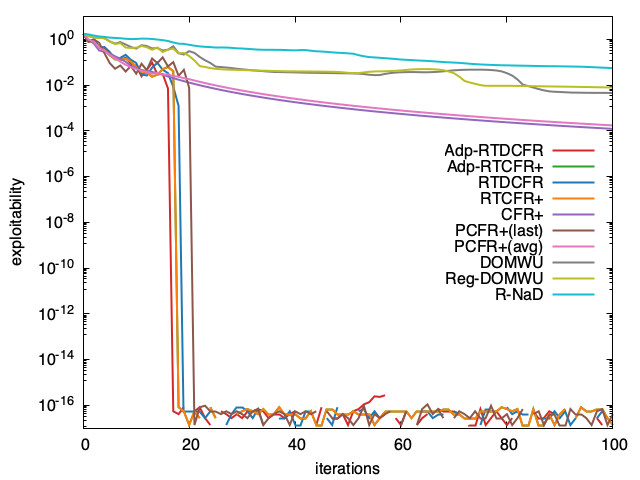}
        \caption{Liar's Dice (6)}
        \label{fig:liar}
    \end{subfigure}

    \caption{Results of EFGs. The numbers in parentheses denote the different rank settings in the respective games: Kuhn poker (3), Leduc poker (3), and Goofspiel (4) indicate the use of 3, 3, and 4 rank cards in poker games, respectively; Liar's dice (6) indicates use a 6-sided die.}
    \label{fig:efg}
\end{figure}

The results presented in Figure \ref{fig:efg} demonstrate that our adaptive method significantly improves the convergence rate, achieving the best performance across all games except for Liar's Dice. All versions of CFR converge to nearly exact NE within just 20 iterations, despite Liar's Dice being the largest-scale game. This is attributed to Liar's Dice (6) possessing a strict NE, where players have a dominant pure strategy, corroborating Theorem~1 in~\citet{cai2023last}. Specifically, the original CFR method also exhibits fast last-iterate convergence in such settings. Furthermore, in Liar's Dice, RTCFRs employ \( T = 1 \), which nullify the RT-term (\( \mu (\sigma^{r,n} - \sigma^{t,n}) = 0 \)) since \( \sigma^{r,n} = \sigma^{t,n} = \sigma^{T+1,n-1} \), reducing RTCFRs to the original CFRs.

Overall, last-iterate convergence methods exhibit faster convergence rates and lower convergence bounds compared to average-iterate convergence methods, especially in Kuhn poker. However, the ability to converge diminishes as game complexity increases, leading to fluctuations and instability, as observed in Leduc poker and Goofspiel. Our adaptive method effectively mitigates this instability by identifying suitable reference strategies and adjusting the regularization weights to balance rapid convergence with exploratory descent directions. This is particularly evident in Goofspiel, where adaptive RTDCFR outperforms other RT algorithms. In contrast, other algorithms exhibit varying performances—R-NaD fails to converge except in Kuhn poker, while PCFR+ shows satisfactory convergence in Goofspiel using the last strategy. These observations suggest an interesting direction for future research: combining the RT technique with prediction to investigate whether it can further enhance convergence performance.

\section{Conclusion and Future Work}

In this paper, we investigated the convergence of the last strategy in regret minimization methods for solving NFGs and EFGs with imperfect information, both theoretically and empirically, and proposed adaptive techniques to enhance this convergence. Initially, we proved that the RT framework enjoys asymptotic last-iterate convergence, although it is limited by certain parameters. To address these limitations, we introduced adaptive methods from three perspectives: selecting reference strategies based on exploitability metrics, balancing performance by adjusting the RT weight to control exploration and exploitation, and efficiently utilizing current regret with a discounting approach. Based on these methods, the proposed algorithm demonstrates significant improvements in empirical performance on both NFG and EFG platforms.

Research on last-iterate convergence has the potential to extend game-theoretic algorithms to deep reinforcement learning, facilitating the approximation of Nash equilibria in large games. Our initial studies focus on simplified games with assumptions of full feedback and rational players, which may not fully reflect real-world conditions. Therefore, we plan to extend our research to include scenarios with bandit feedback and irrational player behaviors in large games. Additionally, prior studies in the RT framework, such as MWU with diverse regularization techniques, motivate the exploration of novel regularization strategies for RM and CFR. Designing effective regularization methods and establishing their convergence properties remain significant open challenges, which we intend to address in subsequent research.


\acks{This research was funded in part by Guangdong Provincial Key Laboratory of Novel Security Intelligence Technologies (No. 2022B1212010005), National Natural Science Foundation of China General Program (No. 62376073), Colleges and Universities Stable Support Project of Shenzhen, China (No. GXWD20220811170225001), and Shenzhen Science and Technology Plan - Major Science and Technology Projects (No. KJZD20230923114213027).}


\newpage

\appendix
\section{Proof of Theorem \ref{thm:best_iterate_convergence_of_rtrms}}\label{app:proof_of_thm:best}
\begin{proof}
    The proof refers to the original work of \citet{meng2023efficient}. We use the connection between the Online Mirror Descent (OMD) formulation and Regret Matching (RM) algorithms to convert RMs into OMD \citep{farina2021faster}. By applying the OMD formulation, we can prove the convergence of RM. 

    First, recall the following lemma:

    \begin{lemma}[Modified from Theorem 2 in \citealt{farina2021faster}]
        The updates for the family of Regret Matching algorithms are equivalent to the following update:
        \begin{equation}
            \theta_i^{t+1,n} \in \argmin_{\theta_i \in \mathbb{R}^d} \left\{ \langle -r_i(\theta_i^{t,n}), \theta_i \rangle + \frac{1}{\eta} D_{\psi}(\theta_i, \theta_i^{t,n}) \right\},
        \end{equation}
        where \(\eta\) is a constant, \(d = |A_i|\), \(r(\theta_i^{t,n}) = \langle \ell_i^{t,n}, \sigma_i^{t,n} \rangle \mathbf{1} - \ell_i^{t,n}\), \(\ell_i^{t,n} = g_i^{t,n} + \mu(\sigma_i^{t,n} - \sigma_i^{r,n})\), \(g_i^{t,n} = -U \sigma_{-i}^{t,n}\) is the loss gradient at iteration \(t\) of the \(n\)-th SCCP, and \(\psi(\cdot)\) is the Euclidean square norm.
    \end{lemma}

    The OMD has a closed-form solution given by:
    \[
    \theta_i^{t+1,n} = \theta_i^{t,n} + \eta r_i(\theta^{t,n}),
    \]
    which satisfies all variants of RM since it only requires the immediate regret results and cumulative regret (Equations \ref{r}, \ref{R}). This maintains the dynamics of the RMs. Thus, the update rule in RM has a closed-form solution:
    \[
    \theta_i^{t+1,n} := \theta_i^{t,n} + \eta r_i(\theta^{t,n}),
    \]
    and for RM+: 
    \[
    \theta_i^{t+1,n} = \max(\theta_i^{t,n} + \eta r_i(\theta_i^{t,n}), 0).
    \]
    For discounted RM, we have:
    \[
    \theta_i^{t+1,n} = \mathbf{w}_i^{t,n} \cdot (\theta_i^{t,n} + \eta r_i(\theta^{t,n})),
    \]
    where
    \[
    \mathbf{w}_i^{t,n}[j] = 
    \begin{cases}
        \frac{t^\alpha}{t^\alpha + 1} & \text{if } (\theta_i^{t,n} + \eta r_i(\theta^{t,n}))[j] > 0, \\
        \frac{t^\beta}{t^\beta + 1} & \text{otherwise}.
    \end{cases}
    \]

    The remaining proof utilizes convex analysis and the technique of rearranging inequalities to obtain the Bregman distance between the saddle point \(\sigma^{*,n}\) and the last strategy \(\sigma^{t,n}\). We omit the detailed derivation here, and refer the reader to Appendix E of \citet{meng2023efficient} for further details. The key result is:
    \begin{equation}
    \sum_{t=1}^T C_1 D_{\psi}(\sigma^{*,n}, \sigma^{t,n}) \leq C_2,
    \label{eq: d saddle last}
    \end{equation}
    where \(C_1 = 2\eta\mu - (\eta C_0)^{2}\), \(C_2 = D_{\psi}(\theta^{1,n,*}, \theta^{1,n}) + \eta \langle -r(\theta^{*,n}), \theta^{1,n} \rangle\), and \(C_0 = 2P^{2} + 3\mu P + P + \mu\), with \(P = \max(|A_1|, |A_2|)\). Here, \(\theta^{1,n,*}\) is the projection of \(\theta^{1,n}\) onto the saddle-point ray.

    By Equation (\ref{eq: d saddle last}) and the Mean Value Theorem, we conclude Theorem \ref{thm:best_iterate_convergence_of_rtrms}. There must exist a time step \(t \leq T\) such that
    \begin{equation}
        \|\sigma^{*,n} - \sigma^{t,n}\|_2 \leq \sqrt{\frac{2C_2}{C_1 T}}.
        \label{eq: sccp converge}
    \end{equation}
\end{proof}

\section{Proof of Lemma \ref{le:relation_of_3_points}} \label{app:proof_of_le:relation_of_3_points}
\begin{proof}
    Let \(\sigma^{*,n}\) be the saddle point of the \(n\)-th SCCP, and \(\sigma^{*}\) be the NE of original game. Then we have:
    \[
    \langle \sigma_i^{*}, \ell_i^{*,n} \rangle \geq \langle \sigma_i^{*,n}, \ell_i^{*,n} \rangle,
    \]
    where \(\ell_i^{*,n} = g_i^{*,n} + \mu(\sigma_i^{*,n} - \sigma_i^{r,n})\) and \(g_i^{*,n} = -U \sigma_{-i}^{*,n}\) is the loss gradient of player \(i\) when the player \(-i\) use strategy \(\sigma_{-i}^{*,n}\).

    Summing over players 1 and 2, we obtain:
    \[
    \sum_{i \in \{1,2\}} \langle \sigma_i^{*}, g_i^{*,n} + \mu(\sigma_i^{*,n} - \sigma_i^{r,n}) \rangle \geq \sum_{i \in \{1,2\}} \langle \sigma_i^{*,n}, g_i^{*,n} + \mu(\sigma_i^{*,n} - \sigma_i^{r,n}) \rangle.
    \]
    
    Rearranging terms, we get:
    \begin{align}
        \langle \sigma^{*}, \mu(\sigma^{*,n} - \sigma^{r,n}) \rangle &\geq \langle \sigma^{*,n}, \mu(\sigma^{*,n} - \sigma^{r,n}) \rangle + 
        \underbrace{\sum_{i \in \{1,2\}}\langle \sigma_i^{*,n}, g_i^{*,n} \rangle - \sum_{i \in \{1,2\}} \langle \sigma_i^{*}, g_i^{*,n} \rangle}_{=\epsilon(\sigma^{*,n})\geq 0}
        \label{eq:convert_to_exp} \\
        &\geq \langle \sigma^{*,n}, \mu(\sigma^{*,n} - \sigma^{r,n}) \rangle.
        \label{eq:q45}
    \end{align}

    Rearranging terms, we obtain Equations (\ref{eq:q1}) and (\ref{eq:q2}):
    \begin{align}
        \langle \sigma^{*,n}, \sigma^{*} - \sigma^{*,n} \rangle &\geq \langle \sigma^{r,n}, \sigma^{*} - \sigma^{r,n} \rangle + \langle \sigma^{r,n}, \sigma^{r,n} - \sigma^{*,n} \rangle,
    \label{eq:q1} \\
         \langle -\sigma^{*}, \sigma^{*} - \sigma^{*,n} \rangle &\geq \langle -\sigma^{*}, \sigma^{*} - \sigma^{r,n} \rangle + \langle -\sigma^{*,n}, \sigma^{r,n} - \sigma^{*,n} \rangle.
    \label{eq:q2}
    \end{align}

    Combining Equations (\ref{eq:q1}) and (\ref{eq:q2}), we have:
    \begin{equation}
        \|\sigma^{*} - \sigma^{*,n}\|_2^2 \leq \|\sigma^{*} - \sigma^{r,n}\|_2^2 - \|\sigma^{r,n} - \sigma^{*,n}\|_2^2.
        \label{eq:bound_of_3_point}
    \end{equation}
    
    which proves Equation~\eqref{eq:saddle_reference_NE}. Assuming \(\sigma^{r,n} \neq \sigma^{*,n} \neq \sigma^{*}\), we apply the identity \(a^2 - b^2 = (a - b)(a + b)\) to obtain:
    \begin{equation}
    \|\sigma^{*} - \sigma^{r,n}\|_2 - \|\sigma^{*} - \sigma^{*,n}\|_2 \geq \frac{\|\sigma^{r,n} - \sigma^{*,n}\|_2^2}{\|\sigma^{*} - \sigma^{r,n}\|_2 + \|\sigma^{*} - \sigma^{*,n}\|_2}. \label{eq:relation_2}
    \end{equation}

    To bound \(\|\sigma^{r,n} - \sigma^{*,n}\|_2^2\), return to Equation~\eqref{eq:convert_to_exp}:
    \begin{align}
    \langle \sigma^{*}, \mu (\sigma^{*,n} - \sigma^{r,n}) \rangle &\geq \langle \sigma^{*,n}, \mu (\sigma^{*,n} - \sigma^{r,n}) \rangle + \epsilon(\sigma^{*,n}). \notag
    \end{align}
    Rearranging terms, we get:
    \begin{equation}
        \epsilon(\sigma^{*,n}) \leq \mu \langle \sigma^{*} - \sigma^{*,n}, \sigma^{*,n} - \sigma^{r,n} \rangle.
        \notag
    \end{equation}
    By the Cauchy-Schwarz inequality:
    \begin{equation}
        \epsilon(\sigma^{*,n}) \leq \mu \|\sigma^{*} - \sigma^{*,n}\|_2 \|\sigma^{*,n} - \sigma^{r,n}\|_2.
        \label{eq:exp_saddle}
    \end{equation}

    By Lemma~\ref{lemma:MS}, there exists a constant \(C > 0\) such that:
    \[
    C \|\sigma^{*} - \sigma^{*,n}\|_2 \leq \epsilon(\sigma^{*,n}) \leq \mu \|\sigma^{*} - \sigma^{*,n}\|_2 \|\sigma^{*,n} - \sigma^{r,n}\|_2.
    \]
    Assuming \(\sigma^{*} \neq \sigma^{*,n}\), we derive:
    \[
    \|\sigma^{*,n} - \sigma^{r,n}\|_2 \geq \frac{C}{\mu}. \label{eq:q35}
    \]
    Substituting into Equation~\eqref{eq:relation_2}, we obtain:
    \[
    \|\sigma^{*} - \sigma^{r,n}\|_2 - \|\sigma^{*} - \sigma^{*,n}\|_2 \geq \frac{C^2}{\mu^2 (\|\sigma^{*} - \sigma^{r,n}\|_2 + \|\sigma^{*} - \sigma^{*,n}\|_2)},
    \]
    proving Equation~\eqref{eq:bound_reference_saddle}. This completes the proof.
    
\end{proof}

\section{Proof of Theorem \ref{theorem: mu converge rate}} \label{app: proof_of_thm:mu_convergence_rate}

\begin{proof}
    Let us analyze the relation between \(\mu\) and the convergence bound in Equation (\ref{eq: sccp converge}). recall the constants \(C_0\), \(C_1\), and \(C_2\) that contain \(\mu\):
    \begin{equation}
        C_0 = 2P^{2} + 3\mu P + P + \mu, \quad P = \max(|A_1|, |A_2|),
        \label{eq:C_0}
    \end{equation}
    and 
    \begin{equation}
        C_1 = 2\eta\mu - (\eta C_0)^2.
        \label{eq:C_1}
    \end{equation}

    For \(C_2\), we expand as follows:
    \begin{align}
        C_2 &= D_{\psi}(\theta^{1,n,*}, \theta^{1,n}) + \eta\langle -r(\theta^{*,n}), \theta^{1,n} \rangle \nonumber\\
        &= \eta \langle \ell^{*,n} - \langle \ell^{*,n}, \sigma^{*,n} \rangle \mathbf{1}, \theta^{1,n} \rangle + D_{\psi}(\theta^{1,n,*}, \theta^{1,n}) \nonumber \\
        &= \eta\|\theta^{1,n}\|_1 \langle \ell^{*,n} - \langle \ell^{*,n}, \sigma^{*,n} \rangle \mathbf{1}, \sigma^{1,n} \rangle + D_{\psi}(\theta^{1,n,*}, \theta^{1,n}) \nonumber \\
        &= \eta\|\theta^{1,n}\|_1 \langle \ell^{*,n}, \sigma^{*,n} - \sigma^{1,n} \rangle + D_{\psi}(\theta^{1,n,*}, \theta^{1,n}) \nonumber \\
        &= \eta\|\theta^{1,n}\|_1 \langle g^{*,n} + \mu(\sigma^{*,n} - \sigma^{r,n}), \sigma^{*,n} - \sigma^{1,n} \rangle + D_{\psi}(\theta^{1,n,*}, \theta^{1,n}) \nonumber \\
        &\leq \eta\|\theta^{1,n}\|_1 (\Omega + 2\mu) + D_{\psi}(\theta^{1,n,*}, \theta^{1,n}) 
        \label{eq:C_2}
    \end{align}
    where \(\Omega = \max_{\sigma_1'} \sigma_1' U \sigma_2 - \min_{\sigma_2'} \sigma_1 U \sigma_2'\). The last inequality follows from the Cauchy-Schwarz inequality and \(\|\sigma - \sigma'\|_2 \leq \sqrt{2}\) for any \(\sigma, \sigma' \in \Delta^n\).

    By substituting \(C_1\) (Equation~\eqref{eq:C_1}) and \(C_2\) (Equation~\eqref{eq:C_2}) into Equation~\eqref{eq: sccp converge}, we derive the following bound:
    \begin{equation}
    \|\sigma^{*,n} - \sigma^{t,n}\|_2 \leq \sqrt{\frac{2}{T}}\cdot \sqrt{\frac{C_2}{C_1}} = \sqrt{\frac{2}{T}} \cdot \sqrt{\frac{\eta \|\theta^{1,n}\|_1 (\Omega + 2\mu) + D_{\psi}(\theta^{1,n,*}, \theta^{1,n})}{2 \eta \mu - (\eta C_0)^2}},
    \label{eq:sccp_bound}
    \end{equation}
    substituting \(\eta=\frac{\mu}{C_0^2}\) and \(C_0=2P^{2} + 3\mu P + P + \mu\), we have:
    \begin{align}
        \|\sigma^{*,n} - \sigma^{t,n}\|_2 &\leq \sqrt{\frac{2}{T}}\cdot \sqrt{\frac{C_3}{\mu}+C_4+\frac{(C_5+\mu(C_6))^2C_7}{\mu^2}}\\
        &=\sqrt{\frac{2}{T}}\cdot \sqrt{\frac{C_5^2C_7}{\mu^2}+\frac{C_3+C_5 C_6 C_7}{\mu}+C_4+C_6^2C_7},
    \end{align} 
    where \(C_3=\|\theta^{1,n}\|_1\Omega\), \(C_4=2\|\theta^{1,n}\|_1\), \(C_5=2P^2+P\), \(C_6=3P+1\), \(C_7=D_{\psi}(\theta^{1,n,*},\theta^{1,n})\).
    This concludes the proof.
\end{proof}

 \section{Proof of Theorem \ref{thm:best_iterate_convergence_of_RTCFRs}} \label{app:proof_of_thm:best_iterate_convergence_of_RTCFRs}

\begin{proof}
    For an extensive-form game, we optimize the behavior strategy using a bottom-up update approach in CFR, similar to RM. Thus, for any information set \(I \in \mathcal{I}\) for any player, based on Equation (\ref{eq: sccp converge}), we have:
    \begin{equation}
        \|\sigma^{*,n}(I) - \sigma^{t,n}(I)\|_2 \leq \sqrt{\frac{2C_2^{\max}}{C_1^{\min} T}},
        \label{eq:bound_local}
    \end{equation}
    where \(C_1^{\min} = \min_{I \in \mathcal{I}_1 \cup \mathcal{I}_2} C_1^I\) and \(C_2^{\max} = \max_{I \in \mathcal{I}_1 \cup \mathcal{I}_2} C_2^I\).

    Let \(\mathcal{Q} \subseteq \mathbb{R}_{\geq 0}^{|\Sigma|}\) denote the sequence-form strategy space, with \(q[\emptyset] = 1\) and \(\sum_{a \in A(I)} q(Ia) = q(pI)\). Define \(M_{\mathcal{Q}} = \max_{q \in \mathcal{Q}} \|q\|_1\), representing the maximum number of information sets with nonzero reach probability under a pure strategy. For an information set \(I \in \mathcal{I}\), let \(q^{t,n}(\triangle_I) \in \mathcal{Q}_{\triangle_I} \subseteq \mathcal{Q}\) denote the sequence-form strategy in the subtree rooted at \(I\). The distance between the saddle-point sequence-form strategy \(q^{*,n}(\triangle_I)\) and the iterate strategy \(q^{t,n}(\triangle_I)\) is defined as:
    \begin{align}
        \|q^{*,n}(\triangle_I) - q^{t,n}(\triangle_I)\| = & q^{*,n}(pI) \Bigg( \|\sigma^{*,n}(I) - \sigma^{t,n}(I)\|_2 \notag\\
        & + \sum_{a \in A(I)} \sum_{I': pI' = Ia} \frac{q^{*,n}(Ia)}{q^{*,n}(pI)} \|q^{*,n}(\triangle_{I'}) - q^{t,n}(\triangle_{I'})\| \Bigg).
    \end{align}

    Using Equation~\eqref{eq:bound_local}, recursively, we bound:
    \[
        \|q^{*,n}(\triangle_I) - q^{t,n}(\triangle_I)\| \leq q^{*,n}(pI) M_{\mathcal{Q}_{\triangle_I}} \sqrt{\frac{2C_2^{\max}}{C_1^{\min} T}},
    \]
    where \(M_{\mathcal{Q}_{\triangle_I}} = \max_{q \in \mathcal{Q}_I} \|q\|_1\). For the entire strategy space (\(q(\emptyset) = 1\)):
    \[
        \|q^{*,n} - q^{t,n}\| \leq M_{\mathcal{Q}} \sqrt{\frac{2C_2^{\max}}{C_1^{\min} T}},
    \]
    yielding the \(O(1/\sqrt{T})\) convergence rate.
\end{proof}



\section{Omitted Details of Experiments} \label{appendix: experiments}
\subsection{Description of the Games} \label{appendix: game description}

\textbf{Kuhn Poker}, introduced in \citet{kuhn1950simplified}, is a simplified form of poker originally using a deck of three cards: King, Queen, and Jack. In its variations, the number of cards can be increased. For instance, in Kuhn(n), the deck consists of \(n\) cards. Each player is dealt one card, with the rest remaining unseen. The betting process involves each player having the option to check, raise, call, or fold, with the player holding the higher card winning the pot. The increased deck size adds complexity and depth to the strategic elements of the game.

\textbf{Leduc Poker}, introduced in \citet{southey2012bayes}, uses a deck of 6 cards: two Kings, two Queens, and two Jacks. Each player is dealt a private card, and there is an additional unrevealed public card. In the first round, Player 2 bets after Player 1 bets. The public card is then revealed, followed by another betting stage. In the showdown stage, the player who has a card matching the rank of the public card wins. If neither player has a matching card, the player with the higher card wins. This game can also be expanded to use any number of cards, such as in the 2n-cards variant Leduc(n), to increase complexity and strategic depth.

\textbf{Goofspiel}, introduced in \citet{ross1971goofspiel}, is an \(n\)-player card game, utilizing \(n + 1\) identical decks, each containing \(k\) cards with values ranging from 1 to \(k\). At the beginning of the game, each player is dealt a full deck as their hand, while the third deck, referred to as the ``prize" deck, is shuffled and placed face down on the board. During each turn, the top card from the prize deck is revealed. Subsequently, each player privately selects a card from their hand to bid for the revealed prize card. The chosen cards are then revealed simultaneously, and the player with the highest card wins the prize. In the event that two or more players reveal cards of equal value, the prize card is split among them. The players' scores are calculated as the sum of the values of the prize cards they have won. In this paper, we set \(n = 2\) and \(k = 4\), with the final reward computed as the difference of scores.

\textbf{Liar's Dice}, introduced in \citet{lisy2015online}, begins with each of the \(n\) players privately rolling a fair \(k\)-sided die. Players then take turns making claims about the results of all dice rolls. The first player starts by announcing any number between 1 and \(k\) and the minimum quantity of dice they believe display that number among all players. Subsequent players can either raise the claim or challenge it by accusing the previous player of lying. A claim is considered higher if it either states a higher number or increases the quantity of dice showing the stated number. If a player challenges a claim and it is proven false, the challenger gains +1 point, and the player who made the false claim loses -1 point. Conversely, if the claim is true, the player who made the claim gains +1 point, and the challenger loses -1 point. In this paper, we set \(n = 2\) and \(k = 6\).

Table \ref{tab: game_info} shows the sizes of the games used in this paper, including the number of information sets, sequences, and leaves. The leaves denote the utility nodes in the sequences of player 1 and player 2 at the terminal stage.

\begin{table}
    \centering
    \caption{Sizes of the games}
    \begin{tabular}{lrrr}
    \toprule
    Game Instance & Information Sets & Sequences & Leaves \\
    \midrule
    Kuhn Poker (3) & 12 & 26 & 30 \\
    Leduc Poker (3) & 288 & 674 & 1116 \\
    Liar's Dice (6) & 24576 & 49142 & 147420 \\
    Goofspiel (4) & 34952 & 42658 & 13824 \\
    \bottomrule
    \end{tabular}
    \label{tab: game_info}
\end{table}

\subsection{Detail Setting in Experiments} \label{appendix: exp setting}
We provide the detailed parameter settings used in our experiments. All parameters were optimized through a search process covering approximately \(1/10\) of the total iterations used in the final experiments.

In Normal-form Games (NFGs), the RT framework includes variants such as RTRM+, RTDRM, adaptive RTRM+, and adaptive RTDRM. We swept for the optimal RT-term weight \(\mu\) and SCCP interval \(T\), with parameters defined as \( \text{par} := (\mu, T) \in \{1, 0.5, 0.1, \dots\} \times \{5, 10, 20, \dots\} \), to find the best configuration. The discount parameters for DRM and adaptive DRM were fixed at \((\alpha, \beta) = (2, 0)\), with further details on parameter selection discussed in Appendix~\ref{app:rtdcfr_test}.

For methods based on MWU, including OMWU, Reg-OMWU, and R-NaD, we tuned the learning rate by applying a logarithmic grid sweep from \(0.01\) to \(10\) across 20 grid points to identify the optimal parameter. For Reg-OMWU and R-NaD, additional regularization weights were optimized similarly to the RT framework, using values in \(\{1, 0.5, 0.1, \dots\}\).

Table \ref{tab:par in matrix} shows the final parameter values used for each algorithm in Figure \ref{fig:matrix_game}.

\begin{table}[ht]
    \centering
    \caption{Hyperparameters used in matrix games for each algorithm and matrix seed}
    \begin{tabular}{lrrrr}
        \toprule
         & Matrix Seed 0 & Matrix Seed 1 & Matrix Seed 2 & Matrix Seed 3 \\
        \midrule
         RTRM+ \((T,\mu)\) & 10, 0.5 & 30, 0.1 & 20, 0.1 & 20, 0.1 \\
         RTDRM \((T,\mu)\) & 5, 0.5 & 20, 0.1 & 20, 0.1 & 20, 0.1 \\
         Adp-RTRM+ \((T, \mu)\) & 20, 0.1 & 40, 0.05 & 20, 0.05 & 30, 0.05 \\
         Adp-RTDRM \((T, \mu)\) & 20, 0.1 & 20, 0.05 & 20, 0.05 & 20, 0.05 \\
         OMWU \((\eta)\) & 0.379 & 0.379 & 0.263 & 0.379 \\
         Reg-OMWU \((\eta, \mu)\) & 0.379, 0.1 & 0.379, 0.1 & 0.263, 0.1 & 0.379, 0.05 \\
         R-NaD \((\eta,T,\mu)\) & 1.128, 30, 0.1 & 1.128, 30, 0.1 & 0.784, 30, 0.05 & 0.784, 20, 0.1 \\
        \bottomrule
    \end{tabular}
    \label{tab:par in matrix}
\end{table}

In EFGs, we use a similar search method to identify efficient parameters for each algorithm. Table \ref{tab:par in EFG} shows the final parameter settings used in Figure \ref{fig:efg}. Additionally, we consider that applying a dilated weight in gradient-based methods may improve convergence in DOMWU. We compare different dilated weights, including the ``all-ones" weight and those from \citet{kroer2015faster}, \citet{kroer2017theoretical}, and \citet{farina2021faster}, applied to Kuhn and Leduc poker. The results in Figure \ref{fig:domwu} show that the ``all-ones" method performs better than the alternatives, so we adopt it in our experiments.

\begin{table}[ht]
    \centering
    \caption{Hyperparameters used in EFGs for each algorithm and game}
    \begin{tabular}{lrrrr}
        \toprule
          & Kuhn (3) & Leduc (3) & Goofspiel (4) & Liar's Dice (6) \\
        \midrule
         RTCFR+ \((T,\mu)\) & 5, 0.1 & 125, 0.001 & 30, 0.005 & 1, 0.1 \\
         RTDCFR \((T,\mu)\) & 5, 0.05 & 125, 0.001 & 20, 0.005 & 1, 0.1 \\
         Adp-RTCFR+ \((T, \mu)\) & 5, 0.05 & 200, 0.01 & 15, 0.1 & 1, 0.01 \\
         Adp-RTDCFR \((T, \mu)\) & 5, 0.05 & 150, 0.01 & 10, 0.1 & 1, 0.01 \\
         DOMWU \((\eta)\) & 0.127 & 0.127 & 0.014 & 0.127 \\
         Reg-DOMWU \((\eta, \mu)\) & 0.127, \(10^{-7}\) & 0.112, 0.001 & 0.127, \(10^{-4}\) & 0.078, 0.01 \\
         R-NaD \((\eta, T, \mu)\) & 0.236, 10, 0.05 & 0.263, 20, 0.1 & 0.029, 10, 0.01 & 0.263, 10, 0.05 \\
        \bottomrule
    \end{tabular}
    \label{tab:par in EFG}
\end{table}

\begin{figure}[htp]
    \centering
    \begin{subfigure}[b]{0.48\linewidth}
        \centering
        \includegraphics[width=\linewidth]{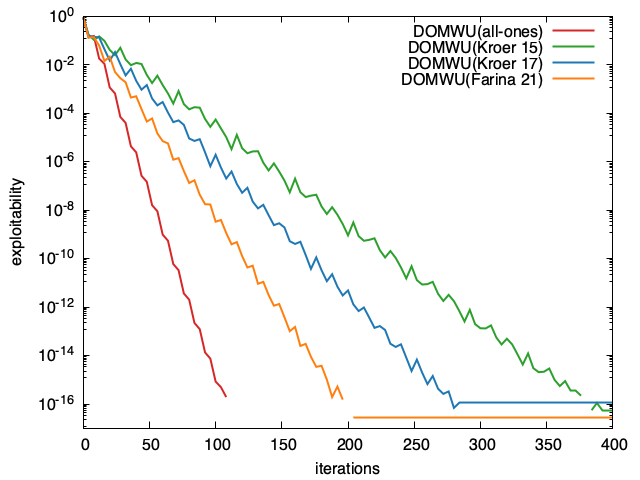}
        \caption{Kuhn Poker (3)}
        \label{fig:domwu-0}
    \end{subfigure}
    \hfill 
    \begin{subfigure}[b]{0.48\linewidth}
        \centering
        \includegraphics[width=\linewidth]{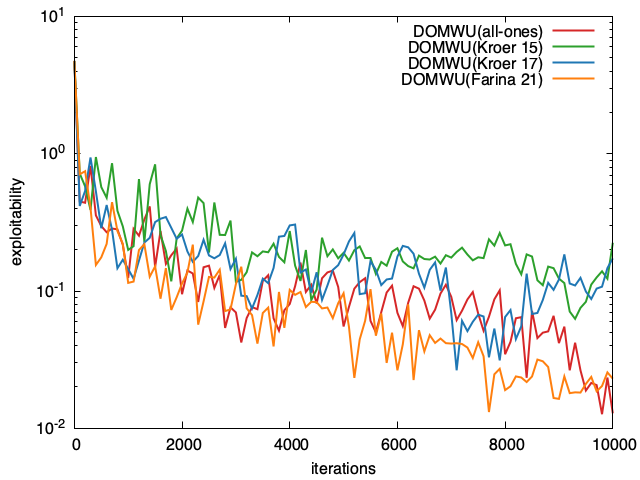}
        \caption{Leduc Poker (3)}
        \label{fig:domwu-1}
    \end{subfigure}

    \caption{Comparison of dilated weights in Kuhn and Leduc poker.}
    \label{fig:domwu}
\end{figure}

\subsection{Experiments on RTDRM and RTDCFR}
\label{app:rtdcfr_test}

In this section, we evaluate the impact of regret weights \((\alpha, \beta) \in \{1.5, 2\} \times \{-\infty, 0, 0.5\}\) on the convergence rate of RTDRM and RTDCFR. Following the experimental setup in Section~\ref{sec:experiments}, we test RTDRM on three \(10 \times 10\) matrix games with seeds \{0, 1, 2\}, with results presented in Figure~\ref{fig:rtdrm}. For EFGs, we evaluate RTDCFR on three EFGs: Kuhn Poker (3), Leduc Poker (3), and Goofspiel (4), with results shown in Figure~\ref{fig:rtdcfr}. The parameters \(T\) and \(\mu\) are selected from Table~\ref{tab:par in matrix} for matrix games and Table~\ref{tab:par in EFG} for EFGs, respectively.

Our results demonstrate that a smaller regret weight \(\alpha = 1.5\) generally underperforms compared to \(\alpha = 2\), exhibiting significant fluctuations and cyclic patterns in convergence and divergence. This behavior aligns with the regret-matching (RM) mechanism, which is driven by positive regret. A smaller \(\alpha\) heavily discounts positive regret, making the cumulative regret overly sensitive to recent updates and leading to cyclic convergence patterns. In contrast, a larger \(\alpha = 2\) promotes more stable convergence. However, excessively large values (e.g., \(\alpha = +\infty, \beta = -\infty\)) cause RTDCFR to degenerate into RTCFR+, which often yields suboptimal performance, highlighting the importance of adaptive regret weighting. For the negative regret discount weight, \(\beta = 0\) generally outperforms other settings, particularly in Kuhn Poker (3). Based on these findings, we select \((\alpha, \beta) = (2, 0)\) for the final experiments in Section~\ref{sec:experiments}.

\begin{figure}[htp]
    \centering
    \begin{subfigure}[b]{0.31\linewidth}
        \centering
        \includegraphics[width=\linewidth]{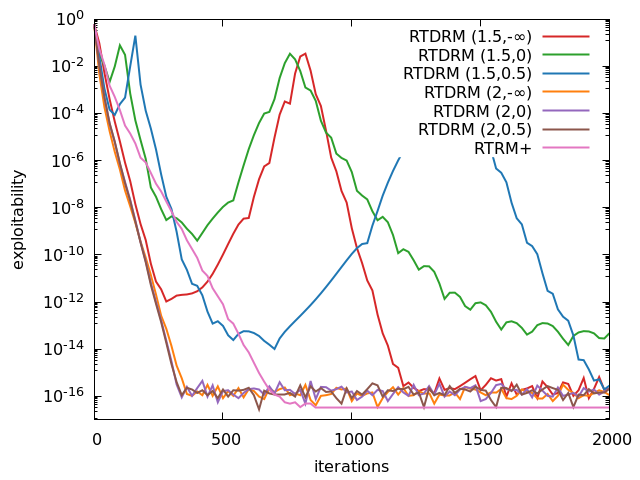}
        \caption{Matrix Game (Seed 0)}
        \label{fig:rtdrm_0}
    \end{subfigure}
    \hfill
    \begin{subfigure}[b]{0.31\linewidth}
        \centering
        \includegraphics[width=\linewidth]{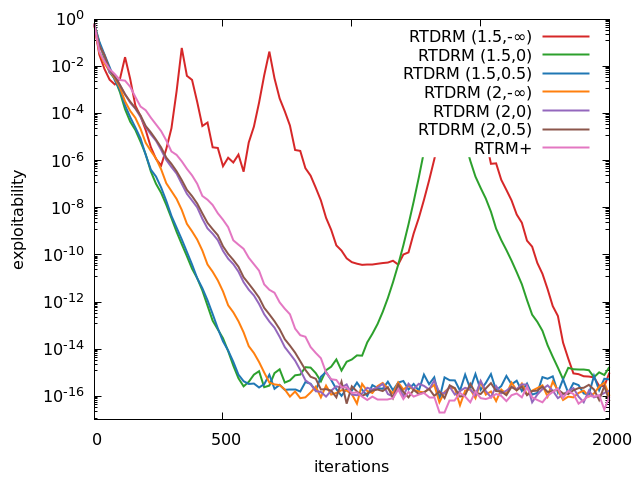}
        \caption{Matrix Game (Seed 1)}
        \label{fig:rtdrm_1}
    \end{subfigure}
    \hfill
    \begin{subfigure}[b]{0.31\linewidth}
        \centering
        \includegraphics[width=\linewidth]{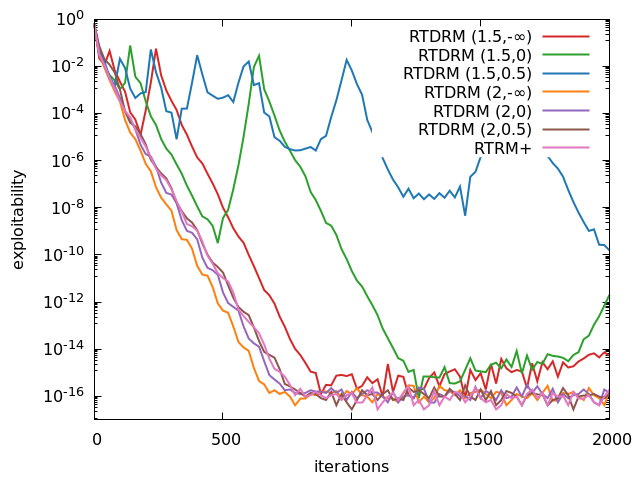}
        \caption{Matrix Game (Seed 2)}
        \label{fig:rtdrm_2}
    \end{subfigure}
    \caption{Convergence of RTDRM with varying \(\alpha\) and \(\beta\) in \(10 \times 10\) matrix games.}
    \label{fig:rtdrm}
\end{figure}

\begin{figure}[htp]
    \centering
    \begin{subfigure}[b]{0.31\linewidth}
        \centering
        \includegraphics[width=\linewidth]{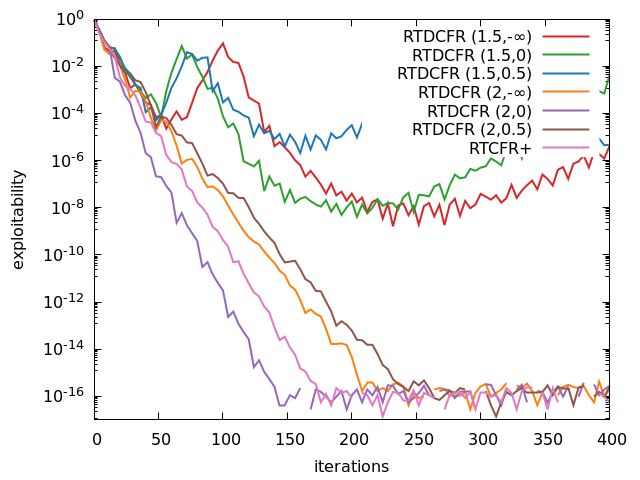}
        \caption{Kuhn Poker (3)}
        \label{fig:rtdcfr_kuhn}
    \end{subfigure}
    \hfill
    \begin{subfigure}[b]{0.31\linewidth}
        \centering
        \includegraphics[width=\linewidth]{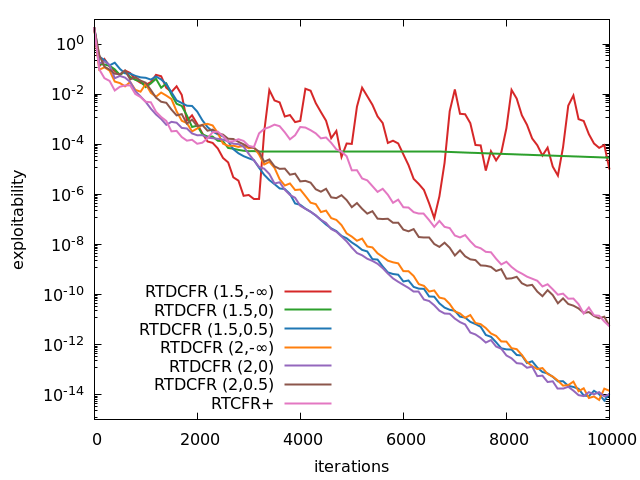}
        \caption{Leduc Poker (3)}
        \label{fig:rtdcfr_leduc}
    \end{subfigure}
    \hfill
    \begin{subfigure}[b]{0.31\linewidth}
        \centering
        \includegraphics[width=\linewidth]{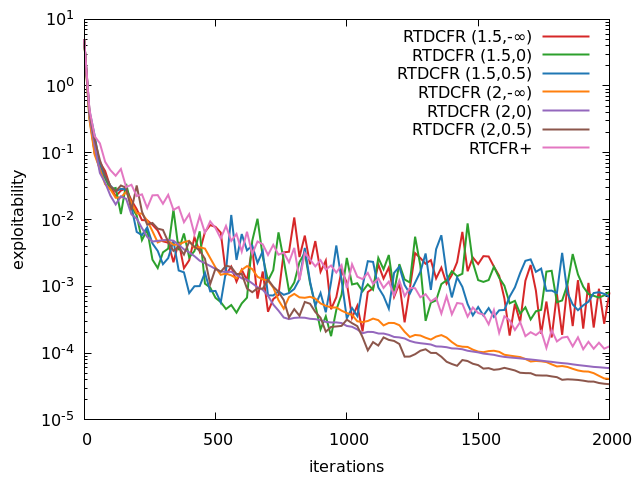}
        \caption{Goofspiel (4)}
        \label{fig:rtdcfr_goofspiel}
    \end{subfigure}
    \caption{Convergence of RTDCFR with varying \(\alpha\) and \(\beta\) in extensive-form games.}
    \label{fig:rtdcfr}
\end{figure}

\vskip 0.2in

\bibliography{jmlr-style-file-master/sample}

\end{document}